\documentclass[journal]{IEEEtran}
\ifCLASSINFOpdf
\else
\fi
\usepackage{amsmath,amsthm}
\usepackage{txfonts}
\newtheorem {theorem}{Theorem}
\usepackage{graphicx}
\usepackage{subfigure}
\usepackage{cleveref}
\usepackage{color}
\usepackage{xcolor}
\Crefname{figure}{Fig.}{Figs.}%
\begin{document}
%
\title{The basic equation for target detection in remote sensing}
%
%
%

\author{Xiurui~Geng,~Luyan~Ji~and~Yongchao Zhao
	\IEEEcompsocitemizethanks{\IEEEcompsocthanksitem X. Geng and Y. Zhao are with the Key Laboratory of Technology in Geo-Spatial information Processing and Application System, Institute of Electronics, Chinese Academy of Sciences, Beijing 100190, China; School of Electronic, Electrical and Communication Engineering, University of Chinese Academy of Sciences, 100049, China \protect\\
		E-mail: gengxr@sina.com
		\IEEEcompsocthanksitem L. Ji is with the Ministry of Education Key Laboratory for Earth System Modeling, Department of Earth System Science, Tsinghua University.}
	\thanks{ Manuscript received ??, ??; revised ??, ??.}}

%
%

\markboth{Journal of ??,~Vol.~??, No.~?, ?~?}%
{Shell \MakeLowercase{\textit{et al.}}: Bare Demo of IEEEtran.cls for IEEE Journals}
%



\maketitle

\begin{abstract}
Our research has revealed a hidden relationship among several basic components, which leads to  the best target detection result. Further, we have proved that the matched filter (MF) is always superior to the constrained energy minimization (CEM) operator, both of which were originally of  parallel importance in the field of target detection for remotely sensed image.
\end{abstract}


%
\IEEEpeerreviewmaketitle

\section{Introduction}
%
%
%
%
\IEEEPARstart{T}{arget} detection is an important research area in the applications of remote sensing, which aims to find an effective detector that can enhance a signal output while suppressing the background. The design of the target detection algorithm must be based on a sufficient understanding of the data. Different criteria will lead to different types of target detection algorithm. Energy and maximum likelihood criteria are the most commonly used \textcolor{black}{methods} for target detection, and they are represented by constrained energy minimization (CEM) \cite{Joseph1993Detection,FARRAND199764,843007,1295199,Liu1999Generalized,Lin2010721} and match filter (MF) \cite{Manolakis2003Hyperspectral,974724,4217133}. CEM \textcolor{black}{was} originally derived from a linearly constrained minimum variance adaptive beam-forming in the field of signal processing. It keeps the output \textcolor{black}{of the CEM detector for the target} as a constant while suppressing output energy \textcolor{black}{in} the background to a minimum level. CEM is widely used for target detection\cite{FARRAND199764, Alam2008Mine,Monaco2014High,5604768} and has been embedded in the most frequently used remote sensing software, \textcolor{black}{(the Environment for Visualizing Images (ENVI))} since Version 4.6. On the other hand, the MF detector, (the minimum variance distortionless response (MVDR) \cite{1449208,1449755,665}) in array processing, is the most commonly utilized technique for various kinds of applications\cite{ParkerWilliams2002446,7071091,1511829,6522163}, and has been included in \textcolor{black}{the} EVNI \textcolor{black}{software in} a very early version. MF is at optimal performance in the Neyman-Pearson sense only when the target and background classes follow multivariate normal distributions with the same covariance matrix, which is an unlikely situation for real-world images.

CEM and MF detectors are designed based on entirely different theoretical foundations. It is interesting to find that their mathematical forms are very similar, except that the MF detector needs the data to be centralized first. The result of CEM can be considered as the inner product of a detector vector and all pixel vectors. Since the inner product is sensitive to the origin position, the performance of CEM will vary with different origins. Yet, as \textcolor{black}{is} known, the data distribution is independent of the selection of origin. From this angle, the CEM \textcolor{black}{is deficient}. MF considers the mean vector as origin, which ``accidentally'' avoids the influence of the origin position on the target detection performance. Yet, is the mean vector the optimal data origin? Do better origins \textcolor{black}{exist other} than zero and mean vector? Clever eye (CE) \cite{Geng201632} provides a strategy to find the best origins, which minimizes the output energy of the background while maintaining a constant response to the signature of interest. Since this is also based on the energy criterion, it can be considered as an optimized version of CEM. 

In this paper, we prove that in order to achieve the optimal target detection performance, the optimal origins must satisfy a basic equation. This basic equation contains four components in target detection: target vector,  mean vector, covariance matrix and origin. Based on this equation, we further prove that MF is always superior to CEM. That is to say, of the two benchmark target detection methods, CEM can now be considered obsolete.

\section{Background} 
In this section, we briefly introduce the formation of \textcolor{black}{the} CEM, MF and CE detectors, and then present the main problem \textcolor{black}{which exists} in CE.

\textcolor{black}{\subsection{CEM}}

CEM \textcolor{black}{was} proposed by Harsanyi in 1993, \textcolor{black}{which exists} originally derived from the linearly constrained minimized variance adaptive beam-forming in the field of digital signal processing. It uses a finite impulse response (FIR) filter to constrain the desired signature by a specific gain while minimizing the filter output energy \cite{Joseph1993Detection}. 

Assume that the observed data matrix is given by $ \mathbf{X}= \left[\mathbf{r}_1,\mathbf{r}_2,\dots,\mathbf{r}_N \right] $ , where  $\mathbf{r}_i=\left[r_{i1},r_{i2},\dots,r_{iL}\right]^T $  for $ 1\le i \le N $ is a sample pixel vector, $ N $ is the total number of pixels, and $ L $ is the number of bands. Suppose that the desired signature $ \mathbf{d} $ is also known. The objective of CEM is to design an FIR linear filter $\mathbf{w}=\left[w_{1},w_{2},\dots,w_{L}\right]^T $ to minimize the filter output \textcolor{black}{energy} subject to the constraint, $ \mathbf{d}^T\mathbf{w}=\sum_{l=1}^{L}d_lw_l=1 $. Then the problem yields 

\begin{equation}
\left\{
\begin{aligned}
&\min \limits_{\mathbf{w}}\frac{1}{N}\left(\sum_{i=1}^{N}y_i^2\right)=\min \limits_{\mathbf{w}}\mathbf{w}^T\mathbf{R}\mathbf{w} \\
&\mathbf{d}^T\mathbf{w}=1 
\end{aligned}\right.,
\label{CEM_obj}
\end{equation}
where $ y_i=\mathbf{w}^T\mathbf{r}_i $ \textcolor{black}{is the output of CEM for the pixel $ \mathbf{r}_i $}, and $ \mathbf{R}=\left(\sum_{i=1}^{N}\mathbf{r}_i\mathbf{r}_i^T\right)/N $, which is firstly referred to as sample correlation matrix by Harsanyi \cite{Joseph1993Detection,FARRAND199764}, and later renamed autocorrelation matrix by some other researchers \cite{Liu1999Generalized,Chang1999Generalized, 1295199, 843007}. In this paper, we will adopt Harsanyi's nomination. The solution to this constrained minimization problem (\ref{CEM_obj}) is the CEM operator, $ \mathbf{w}_{CEM} $, which is given by \cite{Joseph1993Detection}
\begin{equation}
\mathbf{w}_{CEM}=\frac{\mathbf{R}^{-1}\mathbf{d}}{\mathbf{d}^T\mathbf{R}^{-1}\mathbf{d}}.
\label{w_CEM}
\end{equation}


\textcolor{black}{\subsection{MF}}

\textcolor{black}{Though MF stems from the binary hypothesis test, it can also be explained from the perspective of the filter output energy. That is, the MF detector can be considered as the optimal solution of target detection when the data origin is positioned at the mean vector.} The normalized expression of an MF detector can be written as \cite{Manolakis2003Hyperspectral,1295218}

\begin{equation}
\mathbf{w}_{MF}=c_{MF}\mathbf{K}^{-1}\left(\mathbf{d}-\mathbf{m}\right)=\frac{\mathbf{K}^{-1}\left(\mathbf{d}-\mathbf{m}\right)}{\left(\mathbf{d}-\mathbf{m}\right)^T\mathbf{K}^{-1}\left(\mathbf{d}-\mathbf{m}\right)}
\label{w_MF}
\end{equation}
where $ \mathbf{m}=\left(\sum_{i=1}^{N}\mathbf{r}_i\right)/N $  is the mean vector, $ \mathbf{K}=\left[\sum_{i=1}^{N}\left(\mathbf{r}_i-\mathbf{m}\right)\left(\mathbf{r}_i-\mathbf{m}\right)^T\right]/N $  is the covariance matrix, $ c_{MF}= 1/\left[\left(\mathbf{d}-\mathbf{m}\right)^T\mathbf{K}^{-1}\left(\mathbf{d}-\mathbf{m}\right)\right] $ is a scalar.


~

\subsection{CE}
CEM and MF have the same expression except that MF \textcolor{black}{first} requires data centralization. However, this difference leads to a difference in the target detection results. That is to say, the selection of the data origin could directly affect the performance of the detector. CE introduces the data origin as \textcolor{black}{a new} variable in the objective function, \textcolor{black}{and the corresponding optimization model becomes
\begin{equation}
	\left\{
	\begin{aligned}
	&\min \limits_{\mathbf{w},\boldsymbol{\muup}}\mathbf{w}^T\mathbf{R}_{\boldsymbol{\muup}}\mathbf{w} \\
	&\left(\mathbf{d}-\boldsymbol{\muup}\right)^T\mathbf{w}=1 
	\end{aligned}\right.,
	\label{CE_obj}
\end{equation}
where
\begin{equation}
	\mathbf{R}_{\boldsymbol{\muup}}=\frac{1}{N}\left[\sum_{i=1}^{N}\left(\mathbf{r}_i-\boldsymbol{\muup}\right)\left(\mathbf{r}_i-\boldsymbol{\muup}\right)^T\right]=\mathbf{K}+\left(\mathbf{m}-\boldsymbol{\muup}\right)\left(\mathbf{m}-\boldsymbol{\muup}\right)^T.
	\label{Ru}
\end{equation} Apparently, when $ \mathbf{m} = \boldsymbol{\muup}$, $ \mathbf{R}_{\boldsymbol{\muup}} $ is equal to the covariance matrix $ \mathbf{K} $. And when $ \mathbf{m} =\mathbf{0}$, $ \mathbf{R}_{\boldsymbol{\muup}} $ is equal to the sample correlation matrix $ \mathbf{R} $.}

\textcolor{black}{Then, similar to CEM, the corresponding detector can be calculated by firstly fixing $ \boldsymbol{\muup} $ as follows:
\begin{equation}
	\mathbf{w}_{\boldsymbol{\muup}}=\frac{\mathbf{R}_{\boldsymbol{\muup}}^{-1}\left(\mathbf{d}-\boldsymbol{\muup}\right)}{\left(\mathbf{d}-\boldsymbol{\muup}\right)^T\mathbf{R}_{\boldsymbol{\muup}}^{-1}\left(\mathbf{d}-\boldsymbol{\muup}\right).}
	\label{w_ce}
\end{equation}}

\textcolor{black}{From (\ref{w_ce}) we can see that a different $ \boldsymbol{\muup} $ can lead to a different detector. It is interesting to find that when $ \boldsymbol{\muup} = \mathbf{m}$, $ \mathbf{w}_{\boldsymbol{\muup}}=\mathbf{w}_{MF} $, and when $ \boldsymbol{\muup} = \mathbf{0}$, $ \mathbf{w}_{\boldsymbol{\muup}}=\mathbf{w}_{CEM} $. However, both of these may not be the optimal solution for (\ref{CE_obj}).}

\textcolor{black}{Substitute (\ref{w_ce}) into the cost function of (\ref{CE_obj}), and the average filter output energy of CE becomes
\begin{equation}
	f\left(\boldsymbol{\muup}\right)=\frac{1}{\left(\mathbf{d}-\boldsymbol{\muup}\right)^T\mathbf{R}_{\boldsymbol{\muup}}^{-1}\left(\mathbf{d}-\boldsymbol{\muup}\right)}.
	\label{ce_f}
\end{equation}}

\textcolor{black}{Note that the numerator of (\ref{ce_f}) is a constant, so minimizing $ f\left(\boldsymbol{\muup}\right) $ is equivalent to maximizing the denominator. We denote the denominator as $  g\left(\boldsymbol{\muup}\right) $, and then the original optimization problem (\ref{CE_obj}) can be transformed into \cite{Geng201632}}
\begin{equation}
\max \limits_{\boldsymbol{\muup}}g\left(\boldsymbol{\muup}\right)=\left(\mathbf{d}-\boldsymbol{\muup}\right)^T\mathbf{R}_{\boldsymbol{\muup}}^{-1}\left(\mathbf{d}-\boldsymbol{\muup}\right).
\label{g_mu}
\end{equation}
 
It should be noted here that $ g \left(\boldsymbol{\muup}\right) $ corresponds to the reciprocal of the \textcolor{black}{average filter output energy. Since the filter output of the target is fixed to 1, the minimization on the average output energy is equivalent to the minimization on the background output energy. As a result, the} larger the $ g \left(\boldsymbol{\muup}\right) $, the smaller the \textcolor{black}{background} output energy, and thus the better the detector performance \textcolor{black}{\cite{Geng201632}}. \textcolor{black}{Table \ref{Tablecemcemf} tabulates the expressions of the detectors, the filter output of the target, and the corresponding output energy for the three target detectors. It is easy to find their similarity in the mathematical expression. And the only difference is the selection of the data origin. It should be noted here that, when we describe a detector as 'best/optimal' or 'superior' in this study, it means that the filter output energy of the detector is minimum or smaller.}

According to (\ref{g_mu}), $ g \left(\boldsymbol{\muup}\right) $ varies with the origin position $ \boldsymbol{\muup} $, which indicates the target detection result is influenced by the selection of origin. \textcolor{black}{In Ref \cite{Geng201632}, the optimal origin  $ \boldsymbol{\muup}^* $, is solved by applying the gradient ascent method to (\ref{g_mu}). Then the CE detector, $ \mathbf{w}_{\boldsymbol{\muup}^*} $ is calculated accordingly.} The derivative of $ g(\boldsymbol{\muup}) $ with respect to $ \boldsymbol{\muup} $ is expressed as
	\begin{equation}
	\begin{split}
	&g'\left(\boldsymbol{\muup}\right)=-2\mathbf{K}^{-1}\left(\mathbf{d}-\boldsymbol{\muup}\right)-\\
	&\frac{2\left(\mathbf{d}-\boldsymbol{\muup}\right)^T\mathbf{K}^{-1}\left(\mathbf{m}-\boldsymbol{\muup}\right)\left(1+\left(\mathbf{m}-\boldsymbol{\muup}\right)^T\mathbf{K}^{-1}\left(\mathbf{m}-\boldsymbol{\muup}\right)\right)\mathbf{K}^{-1}\left(-\mathbf{m}+2\boldsymbol{\muup}-\mathbf{d}\right)}{\left(1+\left(\mathbf{m}-\boldsymbol{\muup}\right)^T\mathbf{K}^{-1}\left(\mathbf{m}-\boldsymbol{\muup}\right)\right)^2}\\
	&+\frac{2\left(\left(\mathbf{d}-\boldsymbol{\muup}\right)^T\mathbf{K}\left(\mathbf{m}-\boldsymbol{\muup}\right)\right)^2\mathbf{K}^{-1}\left(\mathbf{m}-\boldsymbol{\muup}\right)}{\left(1+\left(\mathbf{m}-\boldsymbol{\muup}\right)^T\mathbf{K}^{-1}\left(\mathbf{m}-\boldsymbol{\muup}\right)\right)^2}.
	\end{split}
	\end{equation}

 However, $ g\left(\boldsymbol{\muup}\right) $ is not a convex function, so the local extreme problem seems inevitable for CE. One possible way to determine the global maximum of $ g\left(\boldsymbol{\muup}\right) $ is by finding out all local maxima, which is apparently an impossible task in practice.

\begin{table}  
\centering	
\caption{The detector, output of target and the average output energy of CEM, MF and CE}
\label{Tablecemcemf} 
	\begin{tabular}{ccccc}  
		\hline  
		Method & Detector &Output of target& Average output energy \\  
		\hline  
		MF  &$\frac{\mathbf{K}^{-1}\left(\mathbf{d}-\mathbf{m}\right)}{\left(\mathbf{d}-\mathbf{m}\right)^T\mathbf{K}^{-1}\left(\mathbf{d}-\mathbf{m}\right)}$ &1  & $\frac{1}{\left(\mathbf{d}-\mathbf{m}\right)^T\mathbf{K}^{-1}\left(\mathbf{d}-\mathbf{m}\right)}$\\
		CEM &$ \frac{\mathbf{R}^{-1}\mathbf{d}}{\mathbf{d}^T\mathbf{R}^{-1}\mathbf{d}} $ &1 & $ \frac{1}{\mathbf{d}^T\mathbf{R}^{-1}\mathbf{d}} $\\
		CE &$ \frac{\mathbf{R}_{\boldsymbol{\muup}}^{-1}\left(\mathbf{d}-\boldsymbol{\muup}\right)}{\left(\mathbf{d}^T-\boldsymbol{\muup}\right)\mathbf{R}_{\boldsymbol{\muup}}^{-1}\left(\mathbf{d}-\boldsymbol{\muup}\right)} $ &1&$ \frac{1}{\left(\mathbf{d}^T-\boldsymbol{\muup}\right)\mathbf{R}_{\boldsymbol{\muup}}^{-1}\left(\mathbf{d}-\boldsymbol{\muup}\right)} $ \\   
		\hline  
	\end{tabular}
\end{table}

\section{Basic equation for target detection}
\subsection{Analytical solution of CE}
The CEM and MF detectors have a similar expression in mathematics, but their theoretical foundations are completely different. As a result, it is hard to theoretically compare their advantages and disadvantages. The \textcolor{black}{CE} algorithm \textcolor{black}{can be regarded as the bridge} between CEM and MF.

 Fortunately, we find that all local maximal points of $ g\left(\boldsymbol{\muup}\right) $ are mathematically equivalent to the solution of a system of linear equations, which can be stated as the following theorem.

\begin{theorem}
	Assume that the target vector $ \mathbf{d} $ is known. To achieve the optimal target detector from the perspective of output energy of the detector, the target vector, $ \mathbf{d} $, data mean vector, $ \mathbf{m} $, the data covariance matrix, $ \mathbf{K} $ and origin, $ \boldsymbol{\muup} $ must satisfy
	\begin{equation}
		\left(\mathbf{d}-\mathbf{m}\right)^T\mathbf{K}^{-1}\left(\mathbf{m}-\boldsymbol{\muup}\right)=1.
		\label{basicEq1}
	\end{equation}
	\label{T1}
\end{theorem}

See Appendix A for details about the proof. (\ref{basicEq1}) is an interesting equation. The general solution set of $ \boldsymbol{\muup} $ in (\ref{basicEq1}) corresponds to all the local maxima of (\ref{g_mu}). Therefore, the optimal origin, $ \boldsymbol{\muup}^* $ can be directly determined by the solution of (\ref{basicEq1}), instead of applying the gradient ascent method as in Ref \cite{Geng201632}. (\ref{basicEq1}) has an infinite number of solutions, with the dimensionality of a solution set equal to the number of bands minus 1. By inspection, it can be found that (\ref{basicEq1}) includes the basic elements in target detection: target, mean vector, covariance matrix and origin. As a result, (\ref{basicEq1})  is named the basic equation of target detection, and also the clever eye (CE) equation.

\subsection{Equivalence between CE and MF}
Theorem \ref{T1} indicates that the optimal origins can be acquired by solving (\ref{basicEq1}). Based on this fact, we further find that though (\ref{basicEq1}) has an infinite number of solutions, they all correspond to the same target detector, which is equivalent to the MF detector. We are now about to establish the equivalence between CE and MF detector.

\begin{theorem}
	The CE detector is equivalent to the MF detector. That is, for any solution of \emph{(\ref{basicEq1})}, $ \boldsymbol{\muup}^* $, there exists a constant $ c $ such that 
	\begin{equation}
		\mathbf{R}_{\boldsymbol{\muup}^*}^{-1}\left(\mathbf{d}-\boldsymbol{\muup}^{*}\right)=c\mathbf{K}^{-1}\left(\mathbf{d}-\mathbf{m}\right)
	\end{equation}
	\label{T2}
\end{theorem}

The proof of Theorem \ref{T2} is demonstrated in Appendix B. This theorem indicates that all local extrema of (\ref{g_mu}) lead to the same detector. Therefore, the local extremum problem of (\ref{g_mu}) does not exist! It is a \textcolor{black}{very} interesting conclusion. Since CE is an optimized version of CEM from the perspective of energy, this equivalence between CE and MF detectors indirectly \textcolor{black}{shows} that MF is always superior to CEM! That is \textcolor{black}{to say}, as one of the two classical algorithms in the field of remote sensing and signal processing, CEM can now be considered redundant. 

\begin{figure}[]\centering
	\subfigure[]{
		\label{fig_gu1}
		\includegraphics[width=0.4\textwidth]{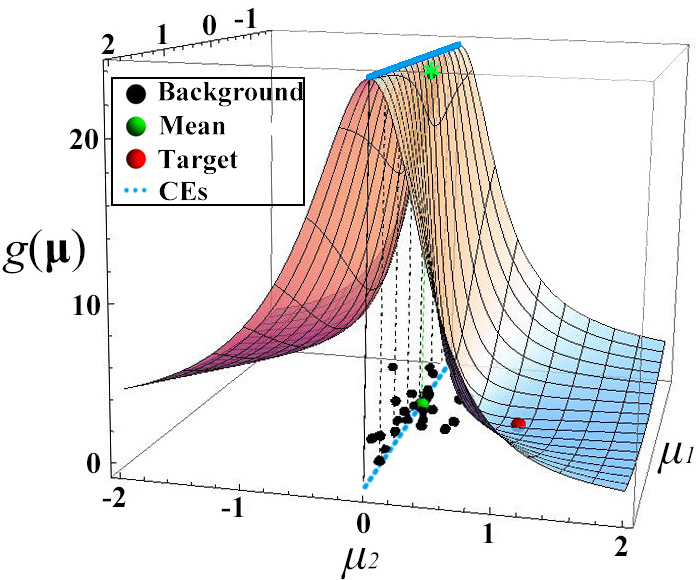}}
	\subfigure[]{
		\label{fig_gu2}
		\includegraphics[width=0.4\textwidth]{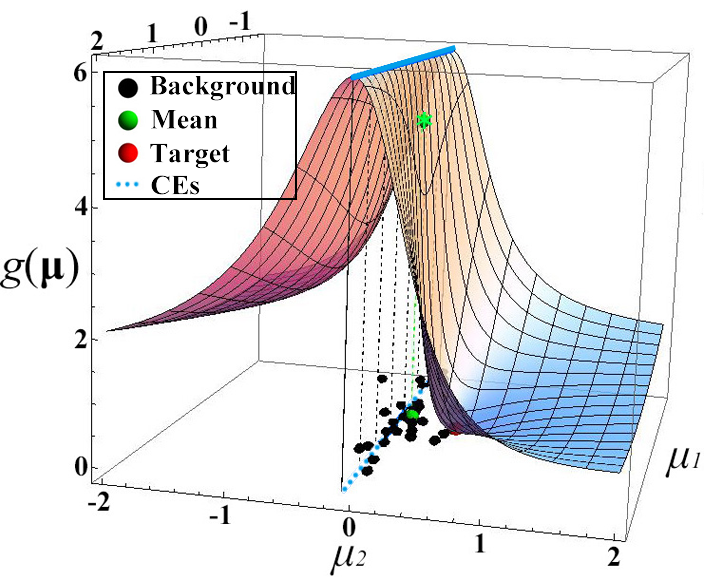}}
	\subfigure[]{
		\label{fig_gu3}
		\includegraphics[width=0.4\textwidth]{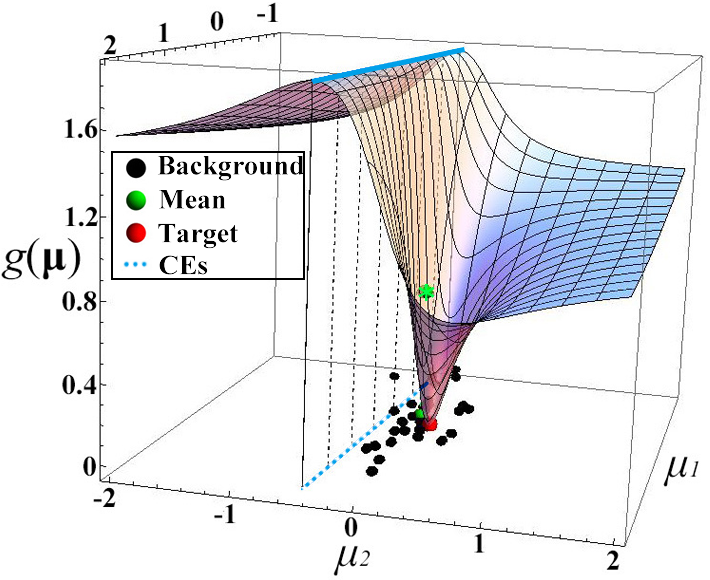}}	
	\caption{ Influence of data origin on the performance of target detection. 30 background points (marked \textcolor{black}{in black}), one target point (marked \textcolor{black}{in red}) and the mean vector (marked \textcolor{black}{in green}) \textcolor{black}{are distributed on the bottom plane}, and the CE \textcolor{black}{line} (marked in \textcolor{black}{blue and dotted}) are \textcolor{black}{also plotted} on the bottom plane. The three sub-figures represent the cases where the target point is far away from (a), close to (b) and within (c) the background points. When $ \mathbf{d} $ is far from the background points, $ g\left(\boldsymbol{\muup}\right) $ has a larger value with a maximum larger than 20, indicating that \textcolor{black}{the separability between target and background is the best}. When $ \mathbf{d} $ is close to background points, the value of $ g\left(\boldsymbol{\muup}\right) $ declines compared with that in (a) (with the maximum value less than 6). Accordingly, \textcolor{black}{the separability between target and background is reduced}. The last case is that the target point lies within the background, which makes it difficult to differentiate the target from the background points. In this case, the value of $ g\left(\boldsymbol{\muup}\right) $ is the lowest with the maximum value less than 2. That is to say, \textcolor{black}{the separability between target and background is the worst}. In all, wherever $ \mathbf{d} $ is, $ g\left(\boldsymbol{\muup}\right) $ can always reach the maxima at the solution of the CE equation.}
	\label{Fig_gu}
\end{figure}

There is an elegant geometry behind (\ref{g_mu}) and (\ref{basicEq1}). Thus, in the following, we will further illustrate the relationship between the objective function in (\ref{g_mu}) and the CE equation in (\ref{basicEq1}) from a geometrical view. First of all, 30 2-dimensional random points, \textcolor{black}{which are scattered in the bottom of Fig. \ref{Fig_gu} and marked in black}, are generated as background by the "randn" function in Matlab. Then, an extra point is manually set as the target \textcolor{black}{and marked in red. In addition, the surface of $ g(\boldsymbol{\muup})$ as the function of $ \boldsymbol{\muup}$ is depicted in Fig.\ref{Fig_gu}. Seemingly, the data origin can be placed in any position on the plane, where the target and background points are located. Yet, different data origins indicate different filter output energies, and thereby different target detection performances. From (\ref{basicEq1}), only those data origins, which satisfy the CE equation, can make the filter output energy reach the minimum (or make $ g(\boldsymbol{\muup})$ maximum), and they are located in the blue dotted straight line in the bottom. On the other hand,}  $ g\left(\boldsymbol{\muup}\right) $ reaches the minimum point at $ \boldsymbol{\muup}=\mathbf{d} $ with $ g\left(\mathbf{d}\right)=0 $. 


From \Cref{Fig_gu} one can notice that the shape of surface $ g\left(\boldsymbol{\muup}\right) $ varies with the target $ \mathbf{d} $. However, the function always reaches the maxima at the solution of the CE equation. Another fact should be pointed out here that the mean vector (green point) is not corresponding to the solution of the CE equation. Yet amazingly, MF has the same target detection performance as CE.

\section{Experiments}
For ease of understanding, we \textcolor{black}{compared} the performance of the three algorithms (CEM, MF, CE) using synthetic and real data. Two metrics, the output energy and correlation, together with the receiver operating characteristic (ROC) curve are adopted in the following experiments.

\begin{figure}[!hbt]\centering	
	\subfigure[]{
		\label{simulated_3d}
		\includegraphics[width=0.48\textwidth]{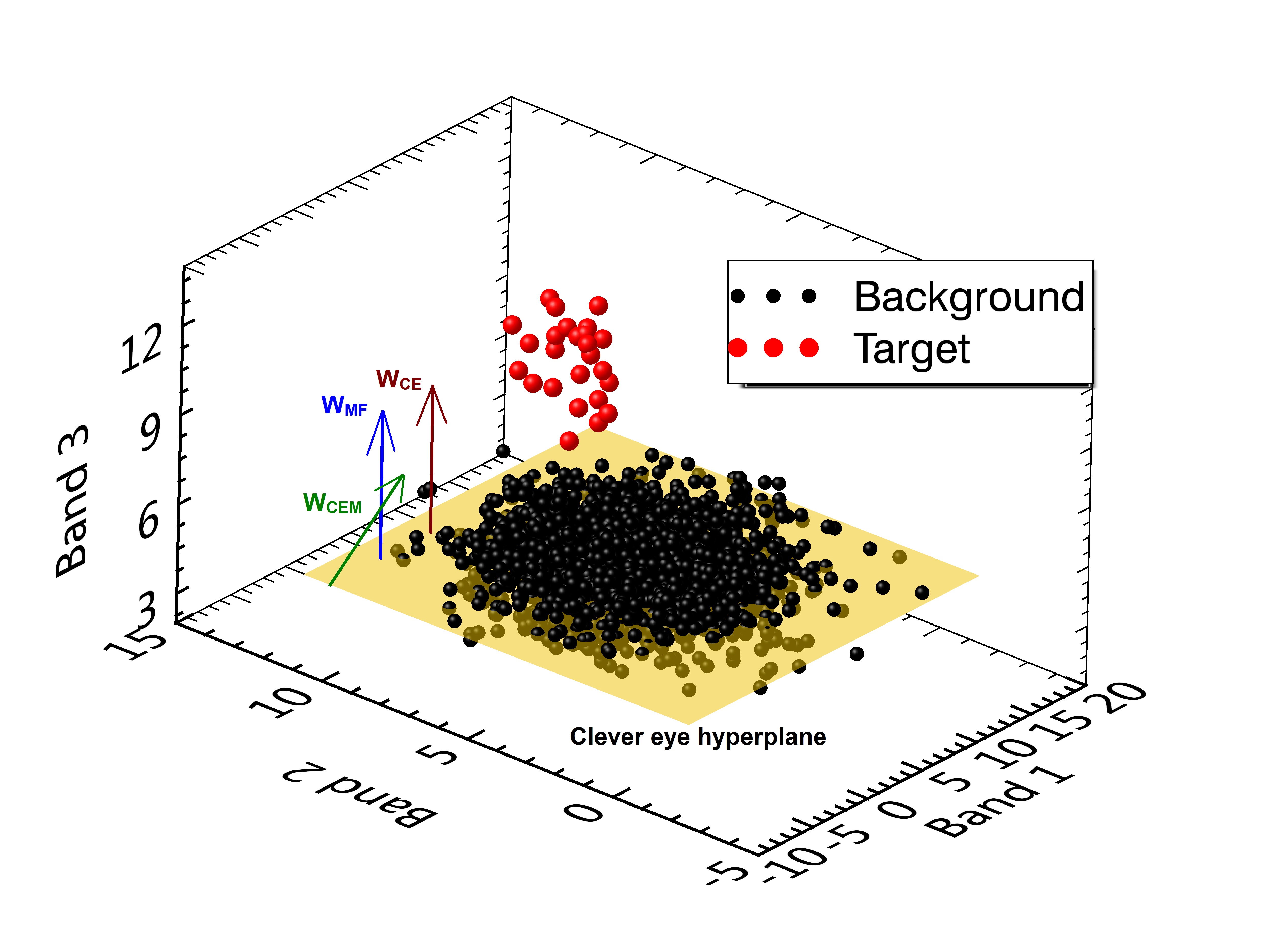}}
	\subfigure[]{
		\label{simulated_true}
		\includegraphics[trim=200mm 110mm 200mm 50mm,clip,width=4.cm]{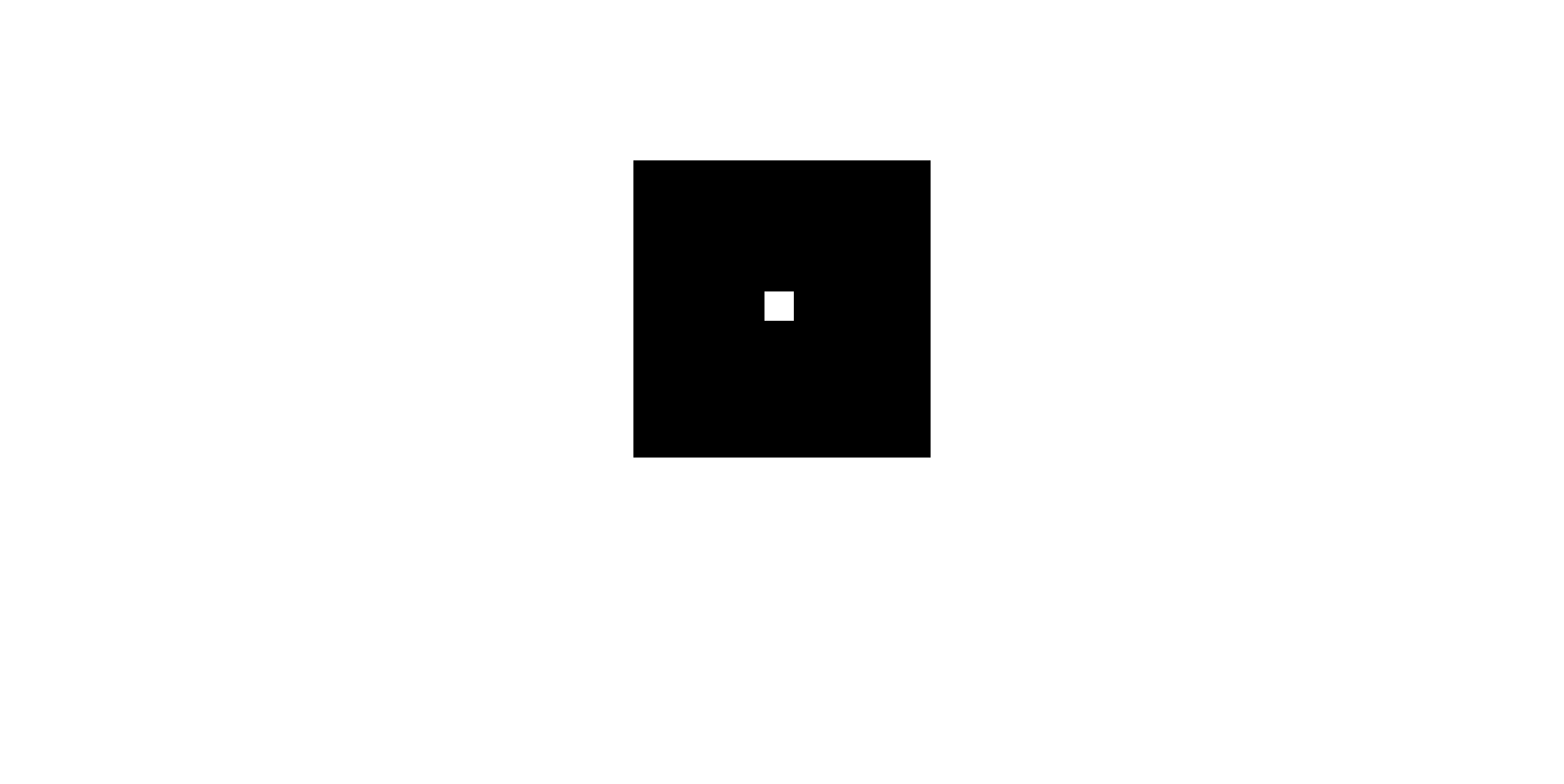}}
	\subfigure[]{
		\label{simulated_cem}
		\includegraphics[trim=200mm 110mm 200mm 51mm,clip,width=4.cm]{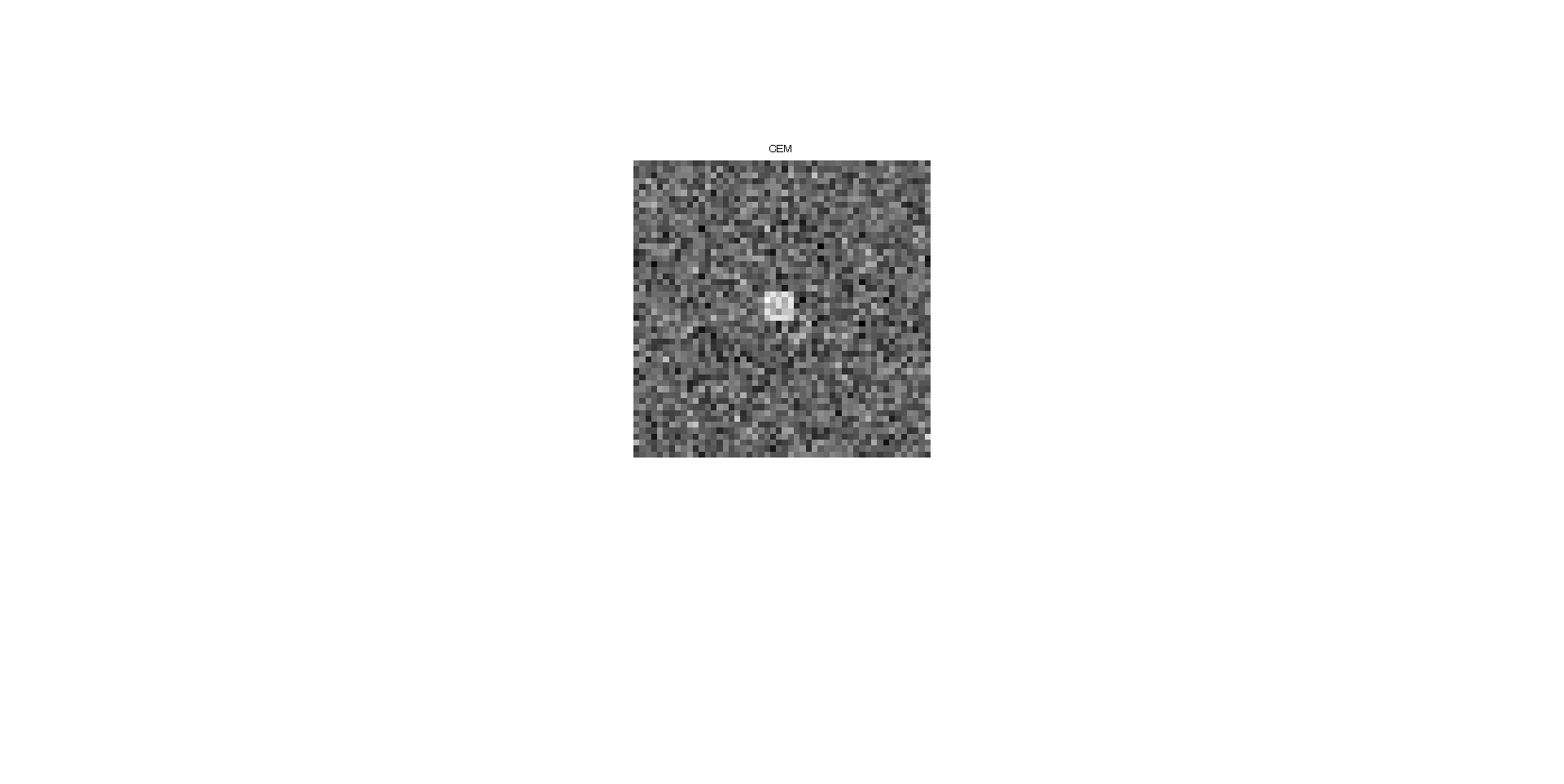}}
	\subfigure[]{
		\label{simulated_mf}
		\includegraphics[trim=200mm 110mm 200mm 50mm,clip,width=4.cm]{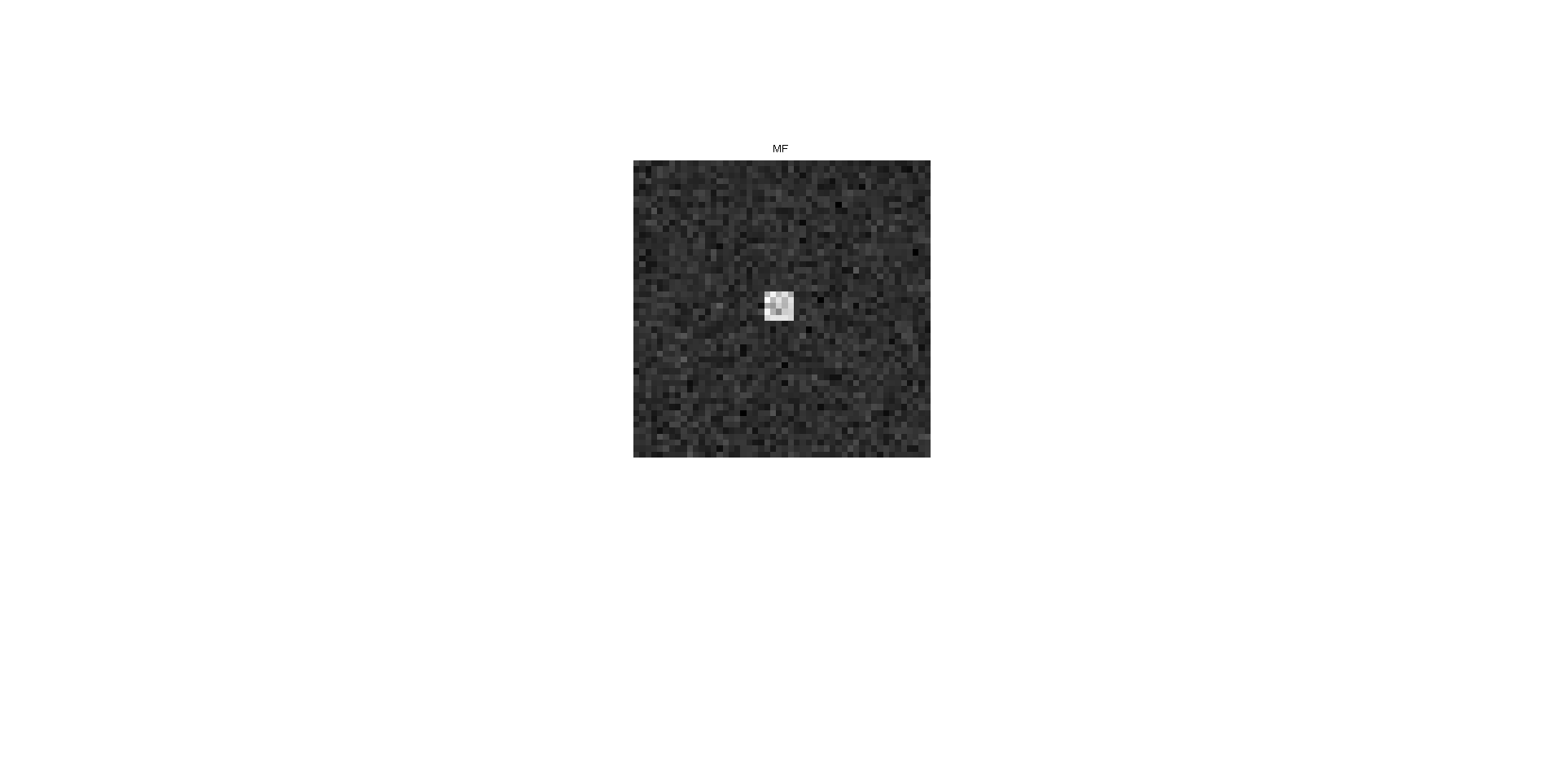}}
	\subfigure[]{
		\label{simulated_ce}
		\includegraphics[trim=200mm 110mm 200mm 50mm,clip,width=4.cm]{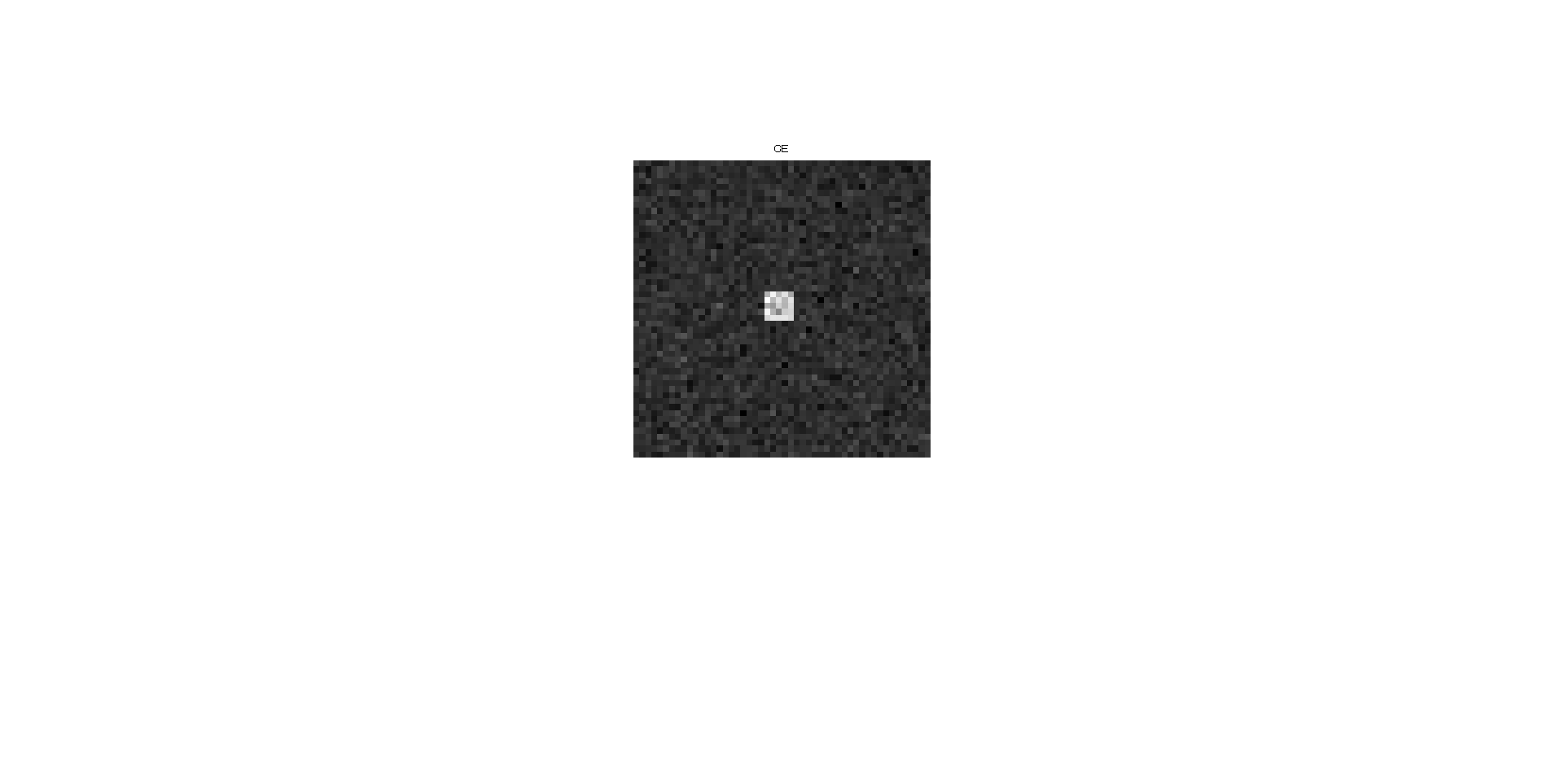}}
	\caption{Distribution of the synthetic image in 2-dimensional feature space (a), true image (b) and the target detection results of CEM (c), MF (d) and CE (e).}
	\label{SimulatedData1}
\end{figure}

\subsection{Experiments with synthetic data}
A synthetic image is designed as follows: 1) generate a $ 50 \times 50 \times 3 $ background image following a 3-dimension normal distribution; 2) add a normally distributed $ 5 \times 5 \times 3 $ target image in the middle of the background image. The distribution of the simulated data set in the three dimensional feature space is shown in Fig. \ref{simulated_3d}. The data cloud of the background is relatively flat with the target points positioned above.

The true image is shown in Fig.\ref{simulated_true} and the target detection outputs of CEM, MF and CE are given in \Cref{simulated_cem,simulated_mf,simulated_ce}. Visually, the performance of CEM is worst, since the contrast between the target and background of CEM is lowest. \textcolor{black}{Also} the output energy of CEM is the largest, as listed in \Cref{OutputEnergy}. The reason \textcolor{black}{that} CE outperforms CEM can be attributed to the fact that CE can search the best data origin for target detection. As revealed by the basic equation in (\ref{basicEq1}), all the CE points are located on a plane, as shown in \Cref{simulated_3d}, and they always lead to the same detector, $ w_{CE} $, which has \textcolor{black}{the} same direction as the MF detector, $ w_{MF} $. However, CEM has a quite different detection direction. Clearly, a better separability between target and background can be achieved along the direction derived from CE and MF.
 
\begin{figure}[!htb]\centering	
	\subfigure[]{
		\label{simulated_acemmf}
		\includegraphics[width=0.23\textwidth]{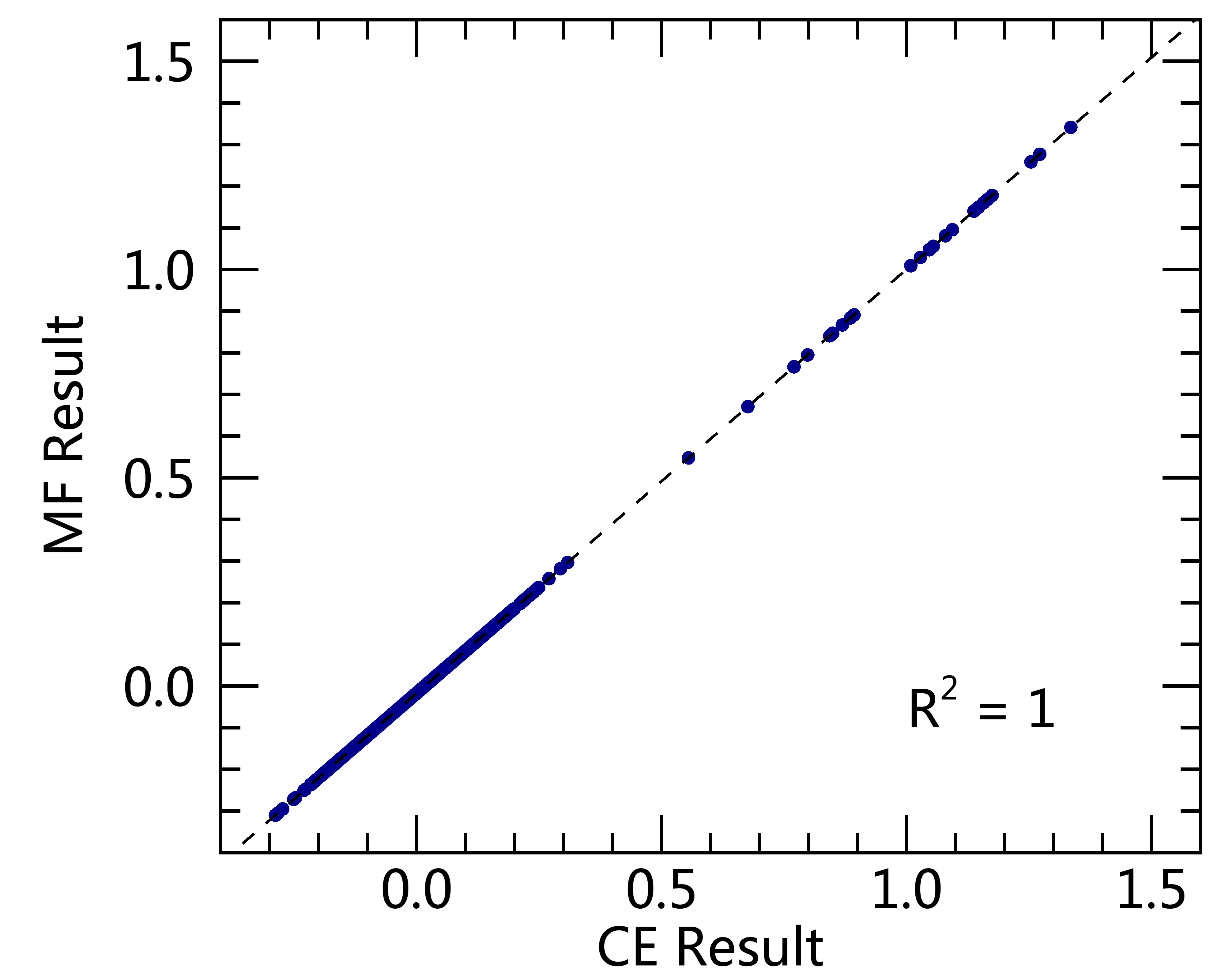}}
	\subfigure[]{
		\label{simulated_cemacem}
		\includegraphics[width=0.23\textwidth]{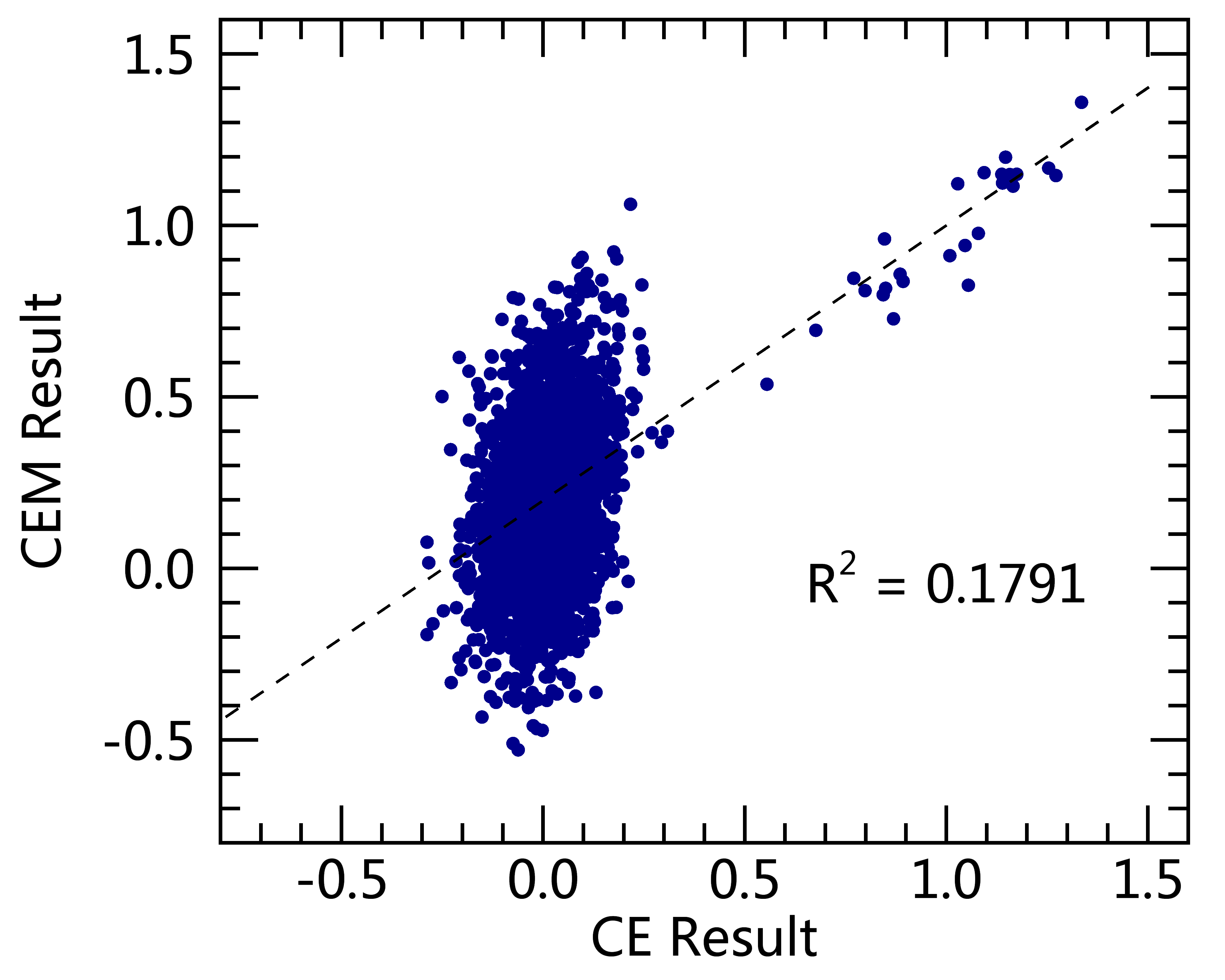}}
	\caption{Comparison of the target detection results of MF and CE (a), and CEM and ranCE(b).}
	\label{SimulatedData2}
\end{figure}

\begin{figure}[!hbt]\centering	
	\includegraphics[width=0.4\textwidth]{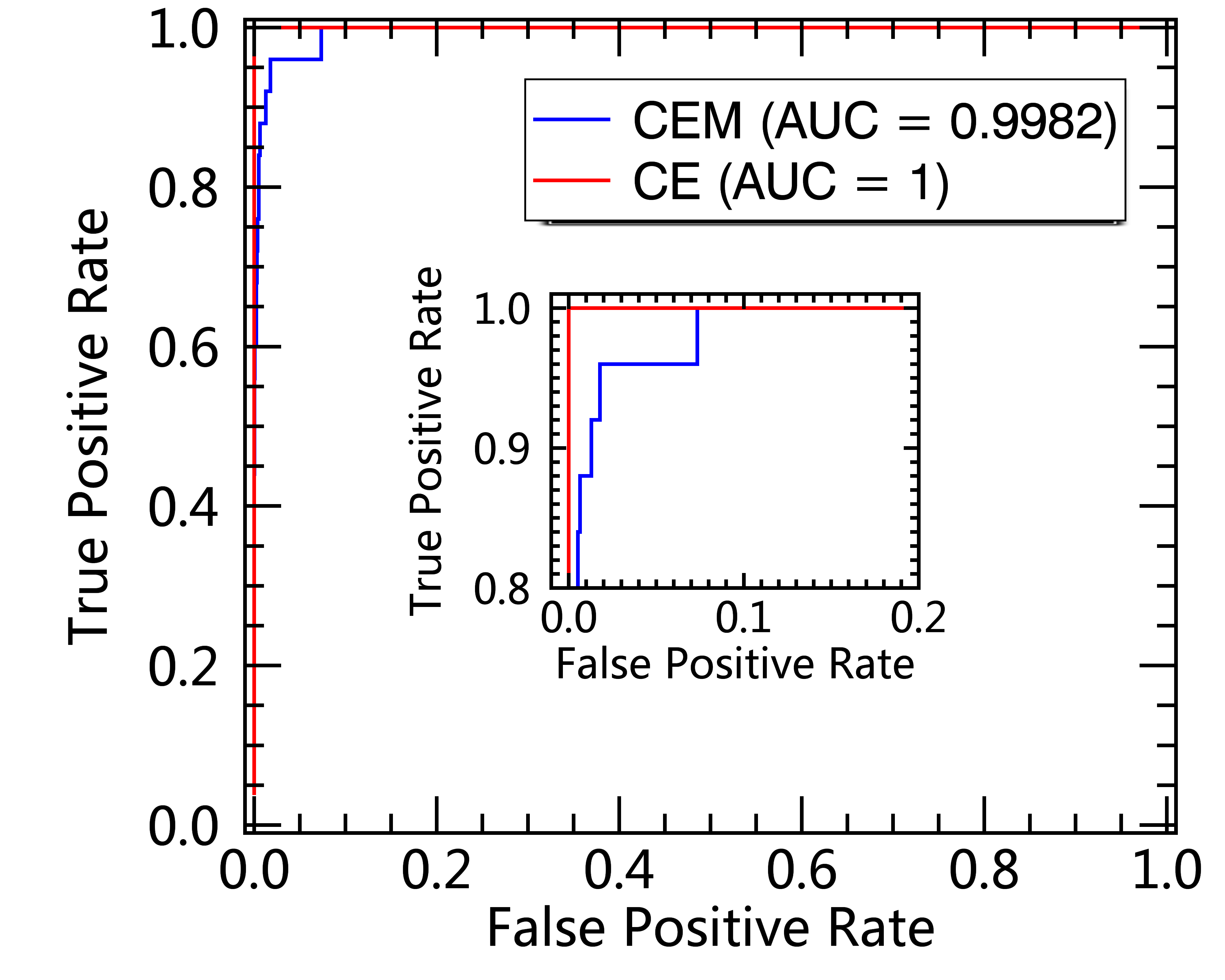}
	\caption{Comparison of the ROC curves for CEM and CE detectors using synthetic data.}
	\label{SimulatedDataroc}
\end{figure}
Yet, it is much harder to compare the results of MF and CE visually, although CE can achieve a lower output energy. Fig.\ref{simulated_acemmf} plots the comparison of the output results of MF and CE. It can be seen that their outputs are completely linearly dependent ($ \rm R^2=1 $)! This result directly indicates the equivalence between CE and MF, which has been proved in Theorem \ref{T2}. On the other hand, the correlation coefficient between the CE and CEM output is very low ($ \rm R^2=0.1791 $) as shown in \Cref{simulated_cemacem}, which infers that changing the data origin can greatly alter the CEM output. The ROC curves of CEM and CE plotted in \Cref{SimulatedDataroc} \textcolor{black}{show} an advantage of CE over CEM in \textcolor{black}{the} target detection of this data set.    

\begin{table}[thp]
	\caption{Output energy of CEM, MF and CE for synthetic and real data.}
	\centering
	\begin{tabular}{cccc}
		\hline
		Method & CEM & MF & CE \\
		\hline
		Synthetic data & 0.1048 & 0.0173 & 0.01701 \\
		Real data  & 9.1895e-04 & 9.1303e-04 & 9.1220e-04\\
		\hline
	\end{tabular}
	\label{OutputEnergy}
\end{table}

\subsection{Experiments with real data}

\begin{figure}[!hbt]\centering	
	\subfigure[]{
		\label{RealData_falsecolor}
		\includegraphics[width=3.5cm]{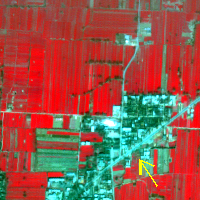}}
	\subfigure[]{
		\label{RealData_groundtruth}
		\includegraphics[width=3.5cm]{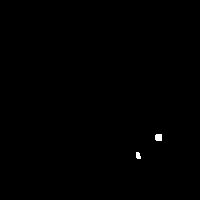}}
	\subfigure[]{
		\label{RealData_d}
		\includegraphics[width=4cm]{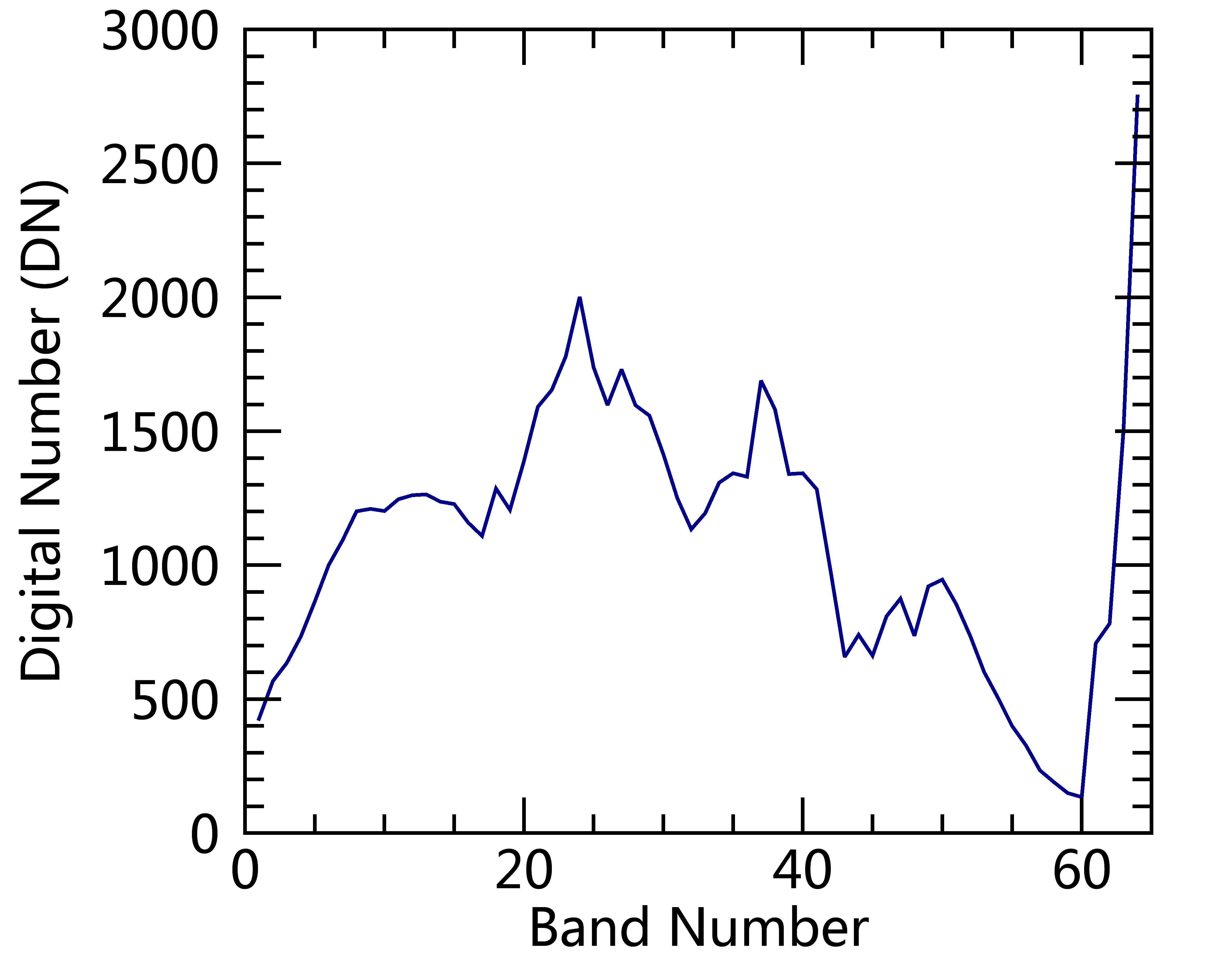}}
	\caption{False color image (R:857.9nm,G:669.6nm,B:559.6nm) of the real data (a), the ground truth image (b) and spectral signature of the target (c).}
	\label{RealData1}
\end{figure}

In this section, we will demonstrate the equivalence between CE and MF by using a real hyperspectral image of $ 200\times200 $  pixels, as shown in Fig.\ref{RealData_falsecolor}. It \textcolor{black}{was} collected by the Operational Modular Imaging Spectrometer-II, which is a hyperspectral imaging system developed by \textcolor{black}{the} Shanghai Institute of Technical Physics, Chinese Academy of Sciences. The data \textcolor{black}{were} acquired in the city of Xi’an, China in 2003\textcolor{black}{, and are} composed of 64 bands from 460nm to 10250 nm \textcolor{black}{at a} spatial resolution \textcolor{black}{of} 3.6m. A small man-made target, marked by an arrow in the scene (\Cref{RealData_falsecolor}) is selected as the target of interest, \textcolor{black}{and the} spectral signature is shown in Fig.\ref{RealData_d}. The manually determined ground truth map based on field survey is shown in \Cref{RealData_groundtruth}.

The target detection results of the three methods are quite similar and no significant differences can be observed visually (\Cref{RealData2}). We present the correlation curves between the outputs of the three detectors in \Cref{RealData_result,RealData_result1}, and it can be seen that the MF output is completely linearly correlated to the CE output ($ \rm R^2=1 $). Again, this result reveals the equivalence between CE and MF.

\begin{figure}[!hbt]\centering	
	\subfigure[]{
		\label{real_CEM}
		\includegraphics[trim=150mm 50mm 150mm 25mm,clip,width=3.5cm]{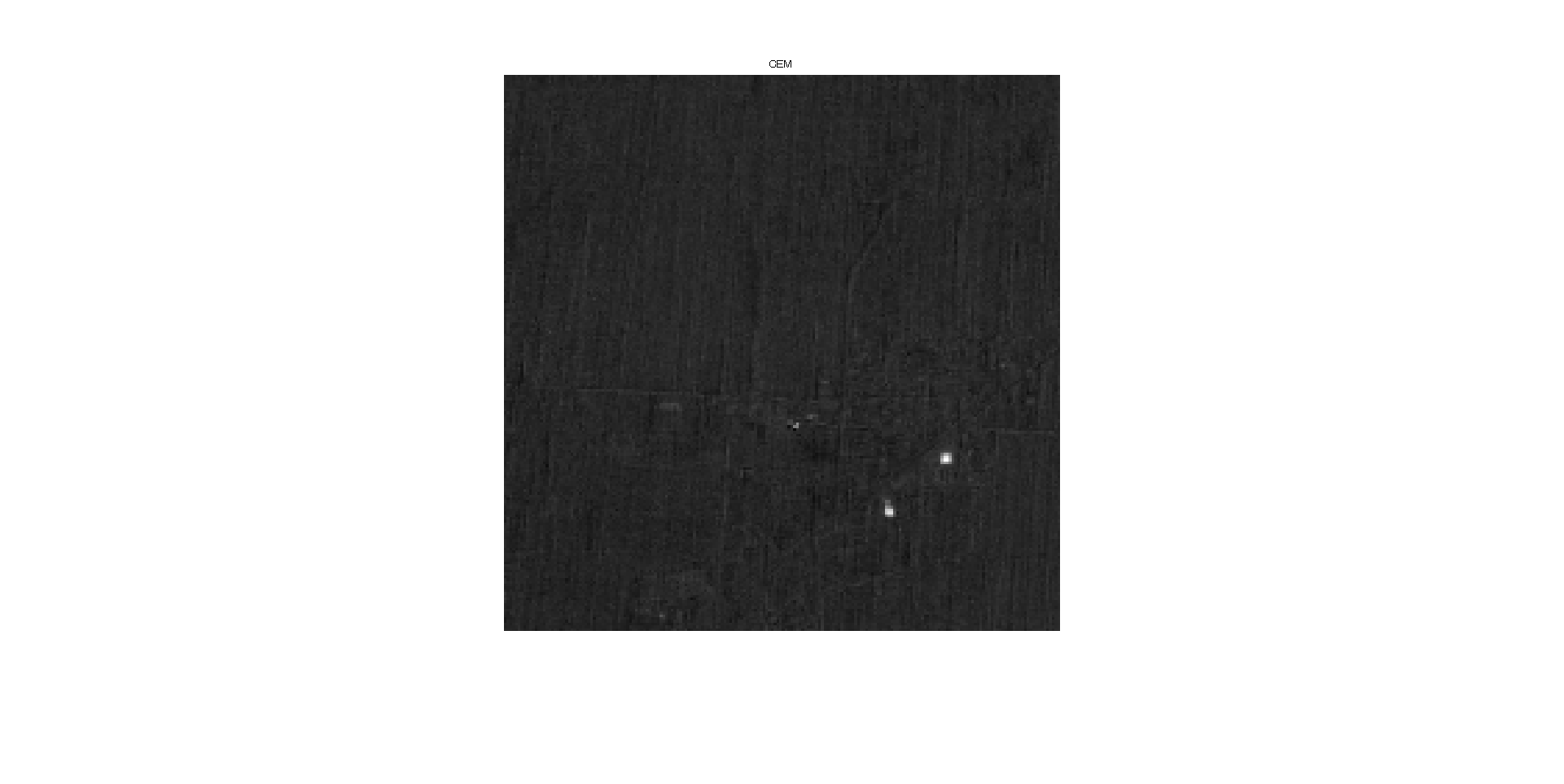}}
	\subfigure[]{
		\label{real_MF}
		\includegraphics[trim=150mm 50mm 150mm 25mm,clip,width=3.5cm]{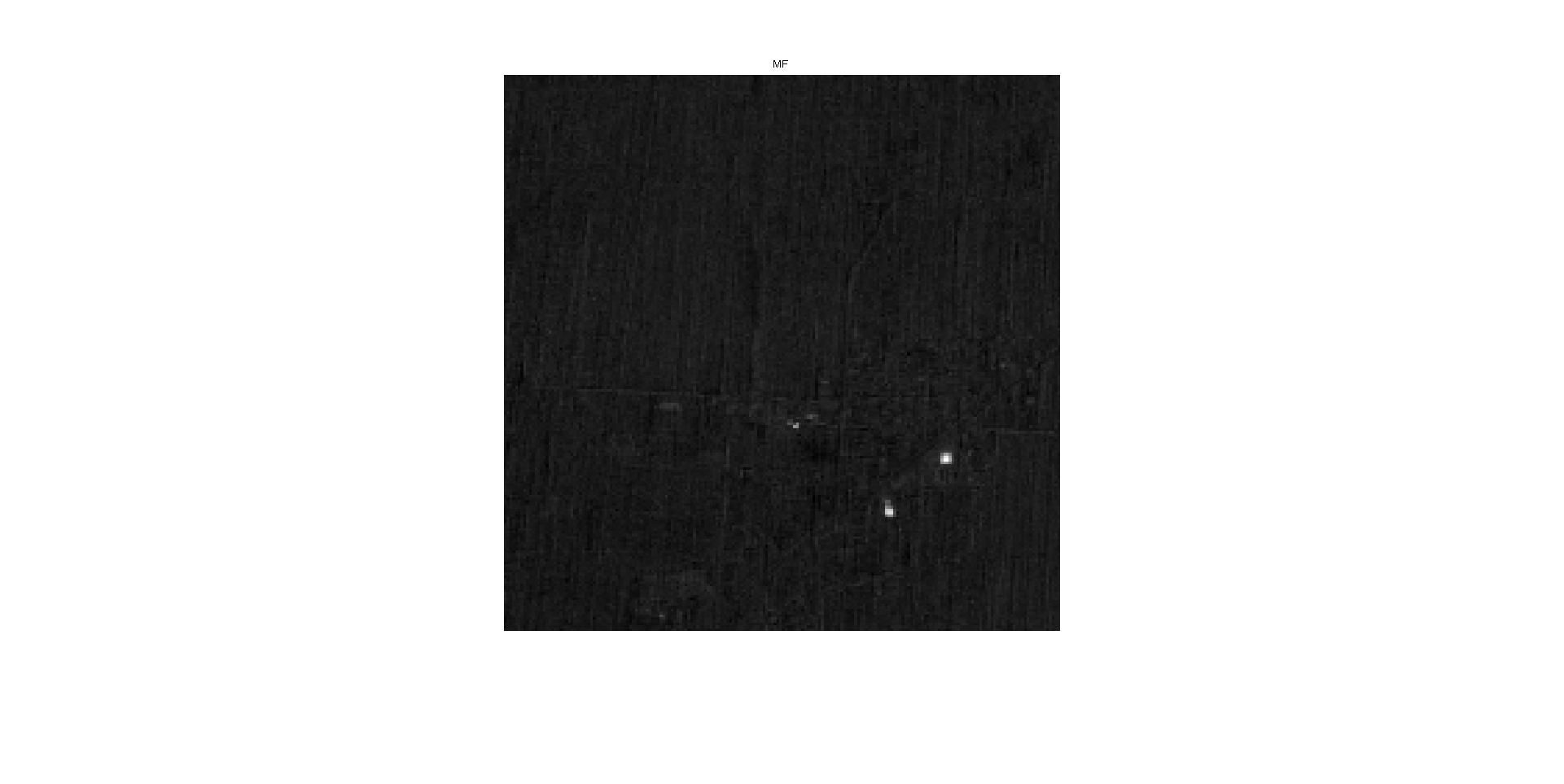}}
	\subfigure[]{
		\label{real_CE}
		\includegraphics[trim=150mm 50mm 150mm 25mm,clip,width=3.5cm]{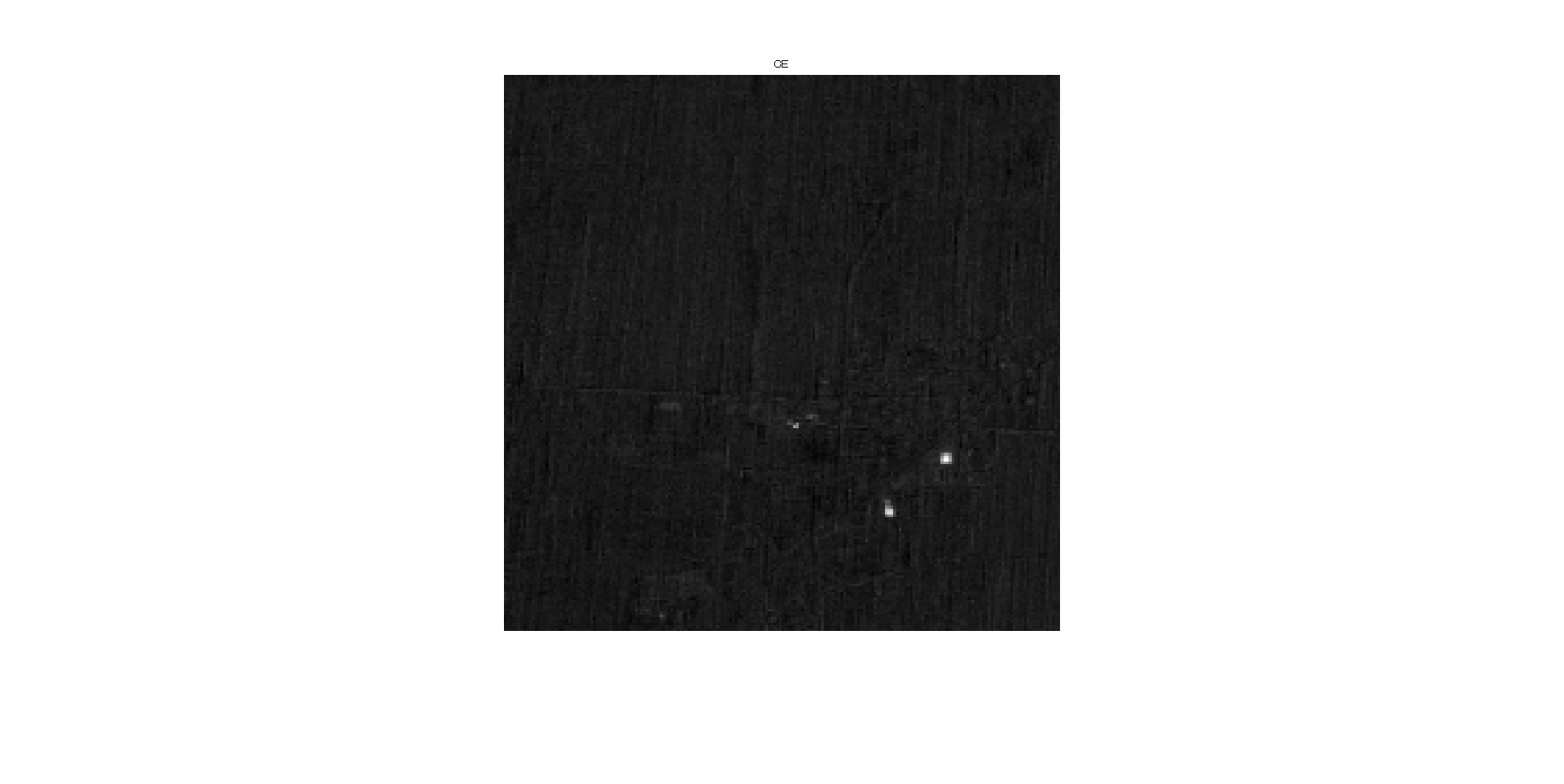}}
	\caption{Target detection results of CEM (a), MF (b) and CE (c).}
	\label{RealData2}
\end{figure}

The target results by CEM and CE \textcolor{black}{are} also very similar with \textcolor{black}{a} correlation as high as $ \rm R^2=0.9926 $. The reason \textcolor{black}{for this} can be explained as follows. According to Ref \cite{Geng201632}, the advantage of MF over CEM decreases with \textcolor{black}{an}
increase \textcolor{black}{in} the number of the bands involved. Since the number of bands for hyperspectral data is usually very high, CEM can derive very similar result as MF and ACEM. That is why CEM has been widely applied for hyperspectral target detection and rarely questioned.  However, the output energy of CEM is still higher than \textcolor{black}{that} of both MF and CE as shown in \Cref{OutputEnergy}. Moreover, according to the ROC curves in \Cref{RealDataroc}, CE still surpasses CEM, though they both have very good detection performance.

\begin{figure}[!hbt]\centering	
	\subfigure[]{
		\label{RealData_result}
		\includegraphics[width=0.23\textwidth]{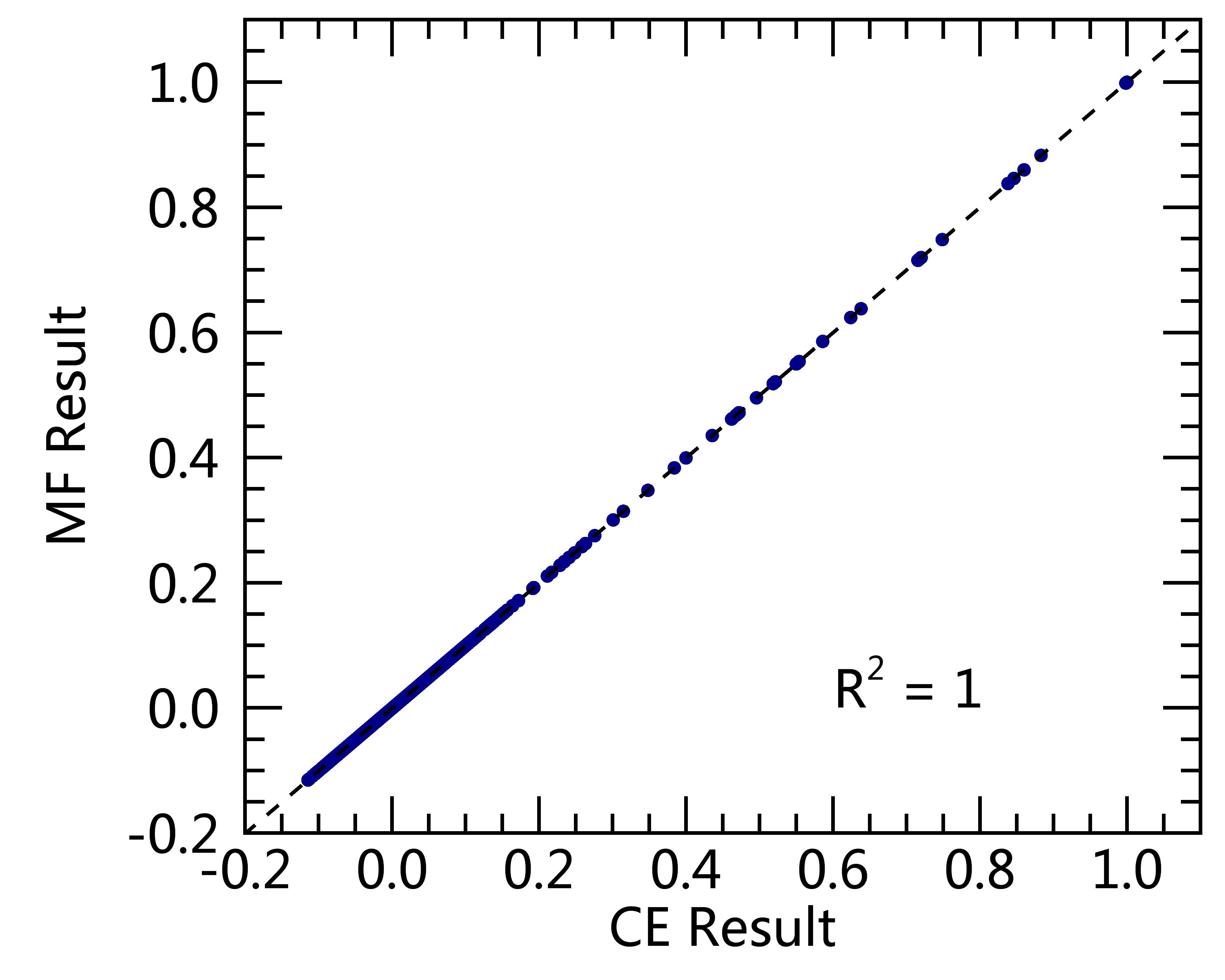}}
	\subfigure[]{
		\label{RealData_result1}
		\includegraphics[width=0.23\textwidth]{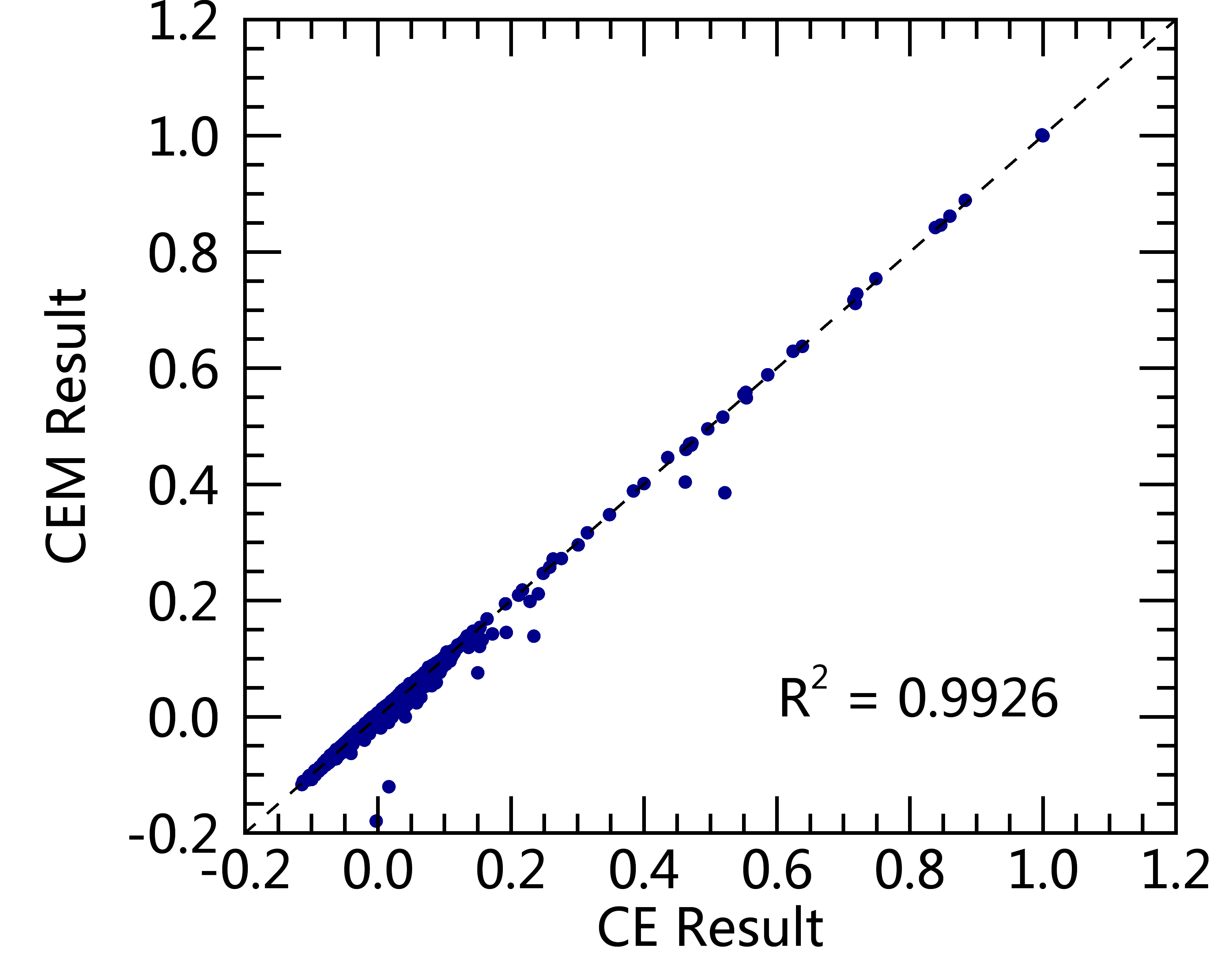}}
	\caption{Comparison of the target detection results of  MF and CE (a), and comparison result between CEM and CE (b) using real data.}
	\label{RealData3}
\end{figure}

\begin{figure}[!hbt]\centering	
      \includegraphics[width=0.4\textwidth]{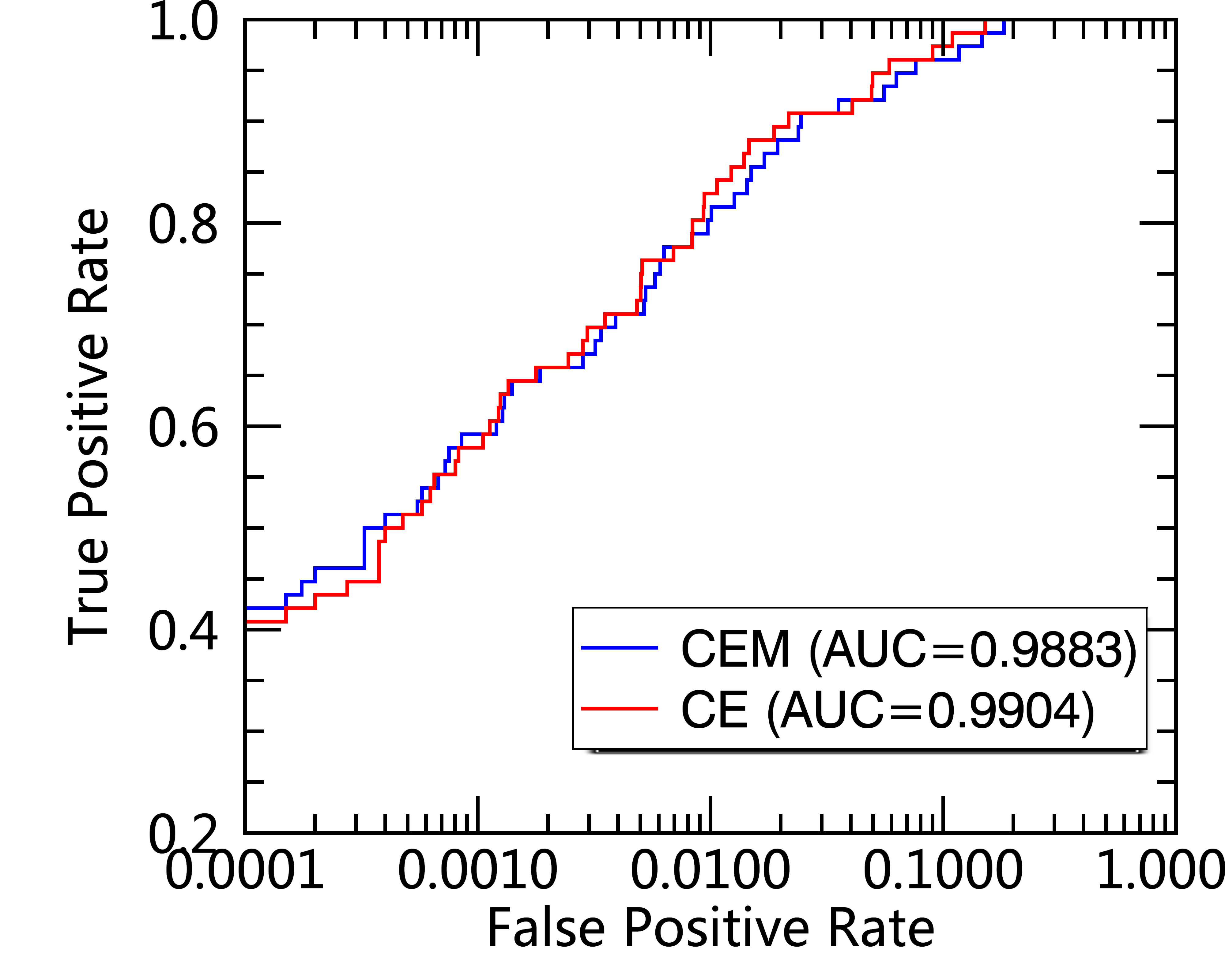}
	\caption{Comparison of the ROC curves for CEM and CE detectors using real data.}
	\label{RealDataroc}
\end{figure}

\section{Conclusion}
MF is the best target detector from the perspective of maximum likelihood, while CEM is the representative one from the perspective of energy. Usually, it is difficult to compare algorithms developed from different criteria, so MF and CEM are considered as two benchmark methods in target detection and both are embedded in the ENVI software, which is one of the most frequently used software packages in the remote sensing community. When the origin is introduced as a variable, it is found that there exists an inherent relationship between the two criteria \cite{Geng201632}. Specifically, we have proved that the basic elements in target detection (including target vector, mean vector, covariance matrix and origin) should follow a fundamental equation to acquire the best detector from the angle of energy. In addition, we prove the equivalence between CE and MF, which indirectly demonstrates that MF is always superior to CEM. Thus, we suggest that the classical target detection CEM should be considered redundant.

On the other hand, the MF detector is not optimal from the perspective of output energy. Yet it is equivalent to the CE detector, which indicates that the energy criterion is not a perfect one. Therefore, we will focus on searching \textcolor{black}{for} a more reasonable criterion for target detection in future \textcolor{black}{studies}.


%

\appendices

\numberwithin{equation}{section}
\section{Proof of Theorem 1}
\setcounter{theorem}{0}
\begin{theorem}
	Assume that the target vector $ \mathbf{d} $ is known. To achieve the optimal target detector from the perspective of output energy, the target vector, $ \mathbf{d} $, data mean vector, $ \mathbf{m} $, the data covariance matrix, $ \mathbf{K} $ and origin, $ \boldsymbol{\muup} $ must satisfy
	\begin{equation}
		\left(\mathbf{d}-\mathbf{m}\right)^T\mathbf{K}^{-1}\left(\mathbf{m}-\boldsymbol{\muup}\right)=1
	\end{equation}
	\label{theorem1}
\end{theorem}
\begin{proof}
	To prove the theorem we first take the derivative of $ g\left(\boldsymbol{\muup}\right) $ with respect to $ \boldsymbol{\muup} $,
	\begin{equation}
		\begin{split}
			&g'\left(\boldsymbol{\muup}\right)=-2\mathbf{K}^{-1}\left(\mathbf{d}-\boldsymbol{\muup}\right)-\\
			&\frac{2\left(\mathbf{d}-\boldsymbol{\muup}\right)^T\mathbf{K}^{-1}\left(\mathbf{m}-\boldsymbol{\muup}\right)\left(1+\left(\mathbf{m}-\boldsymbol{\muup}\right)^T\mathbf{K}^{-1}\left(\mathbf{m}-\boldsymbol{\muup}\right)\right)\mathbf{K}^{-1}\left(-\mathbf{m}+2\boldsymbol{\muup}-\mathbf{d}\right)}{\left(1+\left(\mathbf{m}-\boldsymbol{\muup}\right)^T\mathbf{K}^{-1}\left(\mathbf{m}-\boldsymbol{\muup}\right)\right)^2}\\
			&+\frac{2\left(\left(\mathbf{d}-\boldsymbol{\muup}\right)^T\mathbf{K}\left(\mathbf{m}-\boldsymbol{\muup}\right)\right)^2\mathbf{K}^{-1}\left(\mathbf{m}-\boldsymbol{\muup}\right)}{\left(1+\left(\mathbf{m}-\boldsymbol{\muup}\right)^T\mathbf{K}^{-1}\left(\mathbf{m}-\boldsymbol{\muup}\right)\right)^2}.
		\end{split}
	\end{equation}
	
	After a little algebra, we can have
	\begin{equation}
		\begin{split}
			&g'\left(\boldsymbol{\muup}\right)=\frac{2\left(1+\left(\mathbf{m}-\boldsymbol{\muup}\right)^T\mathbf{K}^{-1}\left(\mathbf{m}-\boldsymbol{\muup}\right)-\left(\mathbf{d}-\boldsymbol{\muup}\right)^T\mathbf{K}^{-1}\left(\mathbf{m}-\boldsymbol{\muup}\right)\right)}{\left(1+\left(\mathbf{m}-\boldsymbol{\muup}\right)^T\mathbf{K}^{-1}\left(\mathbf{m}-\boldsymbol{\muup}\right)\right)^2}\\
			&\cdot\left[-\left(1+\left(\mathbf{m}-\boldsymbol{\muup}\right)^T\mathbf{K}^{-1}\left(\mathbf{m}-\boldsymbol{\muup}\right)\right)\mathbf{K}^{-1}\mathbf{b}\right.\\
			&+\left(1+\left(\mathbf{m}-\boldsymbol{\muup}\right)^T\mathbf{K}^{-1}\left(\mathbf{m}-\boldsymbol{\muup}\right)-\left(\mathbf{d}-\boldsymbol{\muup}\right)^T\mathbf{K}^{-1}\left(\mathbf{m}-\boldsymbol{\muup}\right)\right)\mathbf{K}^{-1}\boldsymbol{\muup}\\
			&\left.+\left(\mathbf{d}-\boldsymbol{\muup}\right)^T\mathbf{K}^{-1}\left(\mathbf{m}-\boldsymbol{\muup}\right)\mathbf{K}^{-1}\mathbf{m}\right].
		\end{split}
		\label{g_dev}
	\end{equation} 
	
	$ g\left(\boldsymbol{\muup}\right) $ reaches its extremum when $ g\left(\boldsymbol{\muup}\right)'=\mathbf{0} $. To see this, let
	\begin{equation}
		1+\left(\mathbf{m}-\boldsymbol{\muup}\right)^T\mathbf{K}^{-1}\left(\mathbf{m}-\boldsymbol{\muup}\right)=\left(\mathbf{d}-\boldsymbol{\muup}\right)^T\mathbf{K}^{-1}\left(\mathbf{m}-\boldsymbol{\muup}\right),
		\label{ea4}
	\end{equation}
	or
	\begin{equation}
		\begin{split}
			&\left(1+\left(\mathbf{m}-\boldsymbol{\muup}\right)^T\mathbf{K}^{-1}\left(\mathbf{m}-\boldsymbol{\muup}\right)-\left(\mathbf{d}-\boldsymbol{\muup}\right)^T\mathbf{K}^{-1}\left(\mathbf{m}-\boldsymbol{\muup}\right)\right)\mathbf{K}^{-1}\boldsymbol{\muup}\\
			&=\left(1+\left(\mathbf{m}-\boldsymbol{\muup}\right)^T\mathbf{K}^{-1}\left(\mathbf{m}-\boldsymbol{\muup}\right)\right)\mathbf{K}^{-1}\mathbf{b}-\left(\mathbf{d}-\boldsymbol{\muup}\right)^T\mathbf{K}^{-1}\left(\mathbf{m}-\boldsymbol{\muup}\right)\mathbf{K}^{-1}\mathbf{m}.
		\end{split}
		\label{ea5}
	\end{equation} 
	
	If (\ref{ea5}) holds, then
	\begin{equation}
		\begin{split}
			&\boldsymbol{\muup}=\frac{\left(1+\left(\mathbf{m}-\boldsymbol{\muup}\right)^T\mathbf{K}^{-1}\left(\mathbf{m}-\boldsymbol{\muup}\right)\right)}{\left(1+\left(\mathbf{m}-\boldsymbol{\muup}\right)^T\mathbf{K}^{-1}\left(\mathbf{m}-\boldsymbol{\muup}\right)-\left(\mathbf{d}-\boldsymbol{\muup}\right)^T\mathbf{K}^{-1}\left(\mathbf{m}-\boldsymbol{\muup}\right)\right)}\mathbf{d}\\
			&-\frac{\left(\mathbf{d}-\boldsymbol{\muup}\right)^T\mathbf{K}^{-1}\left(\mathbf{m}-\boldsymbol{\muup}\right)}{\left(1+\left(\mathbf{m}-\boldsymbol{\muup}\right)^T\mathbf{K}^{-1}\left(\mathbf{m}-\boldsymbol{\muup}\right)-\left(\mathbf{d}-\boldsymbol{\muup}\right)^T\mathbf{K}^{-1}\left(\mathbf{m}-\boldsymbol{\muup}\right)\right)}\mathbf{m}
		\end{split}
		\label{mu}
	\end{equation}
	
	It is easy to observe that the sum of the two coefficients of $ \mathbf{d} $ and $ \mathbf{m}  $ is equal to 1. For simplification , set 
	\begin{equation*}
		\alpha=\frac{\left(1+\left(\mathbf{m}-\boldsymbol{\muup}\right)^T\mathbf{K}^{-1}\left(\mathbf{m}-\boldsymbol{\muup}\right)\right)}{\left(1+\left(\mathbf{m}-\boldsymbol{\muup}\right)^T\mathbf{K}^{-1}\left(\mathbf{m}-\boldsymbol{\muup}\right)-\left(\mathbf{d}-\boldsymbol{\muup}\right)^T\mathbf{K}^{-1}\left(\mathbf{m}-\boldsymbol{\muup}\right)\right)}.
	\end{equation*}
	
	Then, (\ref{mu}) can be rewritten as
	\begin{equation}
		\boldsymbol{\muup}=\alpha\mathbf{d}+\left(1-\alpha\right)\mathbf{m}.
		\label{mu2}
	\end{equation}
	
	It follows that 
	\begin{equation}
		\left(\mathbf{m}-\boldsymbol{\muup}\right)=\alpha\left(\mathbf{m}-\mathbf{d}\right),
		\label{m_mu}
	\end{equation}
	\begin{equation}
		\left(\mathbf{d}-\boldsymbol{\muup}\right)=\left(1-\alpha\right)\left(\mathbf{d}-\mathbf{m}\right).
		\label{d_mu}
	\end{equation}
	
	Substitute (\ref{mu2}), (\ref{m_mu}) and (\ref{d_mu}) into (\ref{mu}), then we can derive $ \boldsymbol{\muup} =\mathbf{d} $. In this case, the objective function $ g\left(\boldsymbol{\muup}\right) =0 $ and reaches the minimum. Yet, in this study, we only focus on the maximum point, so $ \boldsymbol{\muup} $ must satisfy (\ref{ea4}).
	By simplification, (\ref{ea4}) becomes
	\begin{equation}
		\left(\mathbf{d}-\mathbf{m}\right)^T\mathbf{K}^{-1}\left(\mathbf{m}-\boldsymbol{\muup}\right)=1
	\end{equation}
	which is the clever eye equation and therefore the theorem holds.
\end{proof}

\section{Proof of Theorem 2}
\begin{theorem}
	The CE detector is equivalent to the MF detector. That is, for any solution of \emph{(\ref{basicEq1})}, $ \boldsymbol{\muup}^* $, there exists a constant $ c $ such that 
	\begin{equation}
		\mathbf{R}_{\boldsymbol{\muup}^*}^{-1}\left(\mathbf{d}-\boldsymbol{\muup}^{*}\right)=c\mathbf{K}^{-1}\left(\mathbf{d}-\mathbf{m}\right)
		\label{Eq_CE_MF}
	\end{equation}
\end{theorem}

\begin{proof}
	Using Sherman--Morrison formula \cite{Sherman1949Adjustment}, the inverse matrix of $ \mathbf{R}_{\boldsymbol{\mu}} $ can be calculated by 
	\begin{equation}
		\mathbf{R}_{\boldsymbol{\muup}}^{-1}=\left(\mathbf{K}+\left(\mathbf{m}-\mathbf{\boldsymbol{\muup}}\right)\left(\mathbf{m}-\mathbf{\boldsymbol{\muup}}\right)^T\right)^{-1}=\mathbf{K}^{-1}-\frac{\mathbf{K}^{-1}\left(\mathbf{m}-\boldsymbol{\muup}\right)\left(\mathbf{m}-\boldsymbol{\muup}\right)^T\mathbf{K}^{-1}}{1+\left(\mathbf{m}-\boldsymbol{\muup}\right)^T\mathbf{K}^{-1}\left(\mathbf{m}-\boldsymbol{\muup}\right)}.
		\label{R_mu_1}
	\end{equation}
	
	Substitute (\ref{R_mu_1}) into (\ref{Eq_CE_MF}), then we have
	\begin{equation}
		\left(\mathbf{K}^{-1}-\frac{\mathbf{K}^{-1}\left(\mathbf{m}-\boldsymbol{\muup}\right)\left(\mathbf{m}-\boldsymbol{\muup}\right)^T\mathbf{K}^{-1}}{1+\left(\mathbf{m}-\boldsymbol{\muup}\right)^T\mathbf{K}^{-1}\left(\mathbf{m}-\boldsymbol{\muup}\right)}\right)\left(\mathbf{d}-\boldsymbol{\muup}\right)=c\mathbf{K}^{-1}\left(\mathbf{d}-\mathbf{m}\right),
	\end{equation}
	and so 
	\begin{equation}
		\mathbf{K}^{-1}\left(\mathbf{d}-\boldsymbol{\muup}\right)=\frac{\mathbf{K}^{-1}\left(\mathbf{m}-\boldsymbol{\muup}\right)\left(\mathbf{m}-\boldsymbol{\muup}\right)^T\mathbf{K}^{-1}}{1+\left(\mathbf{m}-\boldsymbol{\muup}\right)^T\mathbf{K}^{-1}\left(\mathbf{m}-\boldsymbol{\muup}\right)}\mathbf{K}^{-1}\left(\mathbf{d}-\boldsymbol{\muup}\right)+c\mathbf{K}^{-1}\left(\mathbf{d}-\mathbf{m}\right).
		\label{Eq_CE_MF_2}
	\end{equation}
	
	Apparently, to prove the theorem is equivalent to prove (\ref{Eq_CE_MF_2}). If $ \boldsymbol{\muup} $ satisfies $ g'\left(\boldsymbol{\muup}\right)=\mathbf{0} $, then according to (\ref{g_dev}) we have
	\begin{equation}
		\begin{split}	
			&\mathbf{K}^{-1}\left(\mathbf{d}-\boldsymbol{\muup}\right)=\\
			&\frac{\left(\mathbf{d}-\boldsymbol{\muup}\right)^T\mathbf{K}^{-1}\left(\mathbf{m}-\boldsymbol{\muup}\right)\left(1+\left(\mathbf{m}-\boldsymbol{\muup}\right)^T\mathbf{K}^{-1}\left(\mathbf{m}-\boldsymbol{\muup}\right)\right)\mathbf{K}^{-1}\left(-\mathbf{m}+2\boldsymbol{\muup}-\mathbf{d}\right)}{\left(1+\left(\mathbf{m}-\boldsymbol{\muup}\right)^T\mathbf{K}^{-1}\left(\mathbf{m}-\boldsymbol{\muup}\right)\right)^2}\\
			&+\frac{\left(\left(\mathbf{d}-\boldsymbol{\muup}\right)^T\mathbf{K}\left(\mathbf{m}-\boldsymbol{\muup}\right)\right)^2\mathbf{K}^{-1}\left(\mathbf{m}-\boldsymbol{\muup}\right)}{\left(1+\left(\mathbf{m}-\boldsymbol{\muup}\right)^T\mathbf{K}^{-1}\left(\mathbf{m}-\boldsymbol{\muup}\right)\right)^2},
		\end{split}
	\end{equation}
	which can be further turned into
	\begin{equation}
		\begin{split}
			\mathbf{K}^{-1}\left(\mathbf{d}-\boldsymbol{\muup}\right)&=\frac{\mathbf{K}^{-1}\left(\mathbf{m}-\boldsymbol{\muup}\right)\left(\mathbf{m}-\boldsymbol{\muup}\right)^T\mathbf{K}^{-1}}{1+\left(\mathbf{m}-\boldsymbol{\muup}\right)^T\mathbf{K}^{-1}\left(\mathbf{m}-\boldsymbol{\muup}\right)}\mathbf{K}^{-1}\left(\mathbf{d}-\boldsymbol{\muup}\right)\\
			&+\frac{\left(\mathbf{d}-\boldsymbol{\muup}\right)^T\mathbf{K}^{-1}\left(\mathbf{m}-\boldsymbol{\muup}\right)\mathbf{K}^{-1}\left(\mathbf{d}-\boldsymbol{\muup}\right)}{1+\left(\mathbf{m}-\boldsymbol{\muup}\right)^T\mathbf{K}^{-1}\left(\mathbf{m}-\boldsymbol{\muup}\right)}\\
			&-\frac{\left(\left(\mathbf{d}-\boldsymbol{\muup}\right)^T\mathbf{K}^{-1}\left(\mathbf{m}-\boldsymbol{\muup}\right)\right)^2\mathbf{K}^{-1}\left(\mathbf{m}-\boldsymbol{\muup}\right)}{\left(1+\left(\mathbf{m}-\boldsymbol{\muup}\right)^T\mathbf{K}^{-1}\left(\mathbf{m}-\boldsymbol{\muup}\right)\right)^2},	
		\end{split}
		\label{Eq_CE_MF_3}
	\end{equation}
	
	Considering (\ref{Eq_CE_MF_2}) and (\ref{Eq_CE_MF_4}), to prove the theorem we need to verify the equivalence of the right sides of both equations as follows,
	\begin{equation}
		\begin{split}	 
			&\frac{\left(\mathbf{d}-\boldsymbol{\muup}\right)^T\mathbf{K}^{-1}\left(\mathbf{m}-\boldsymbol{\muup}\right)\mathbf{K}^{-1}\left(\mathbf{d}-\boldsymbol{\muup}\right)}{1+\left(\mathbf{m}-\boldsymbol{\muup}\right)^T\mathbf{K}^{-1}\left(\mathbf{m}-\boldsymbol{\muup}\right)}\\
			&-\frac{\left(\left(\mathbf{d}-\boldsymbol{\muup}\right)^T\mathbf{K}^{-1}\left(\mathbf{m}-\boldsymbol{\muup}\right)\right)^2\mathbf{K}^{-1}\left(\mathbf{m}-\boldsymbol{\muup}\right)}{\left(1+\left(\mathbf{m}-\boldsymbol{\muup}\right)^T\mathbf{K}^{-1}\left(\mathbf{m}-\boldsymbol{\muup}\right)\right)^2}\\
			&=c\mathbf{K}^{-1}\left(\mathbf{d}-\mathbf{m}\right)
		\end{split}
		\label{Eq_CE_MF_4}
	\end{equation}
	
	By theorem \ref{theorem1}, (\ref{Eq_CE_MF_4}) can be simplified as
	\begin{equation}
		\mathbf{K}^{-1}\left(\mathbf{d}-\mathbf{m}\right)=c\mathbf{K}^{-1}\left(\mathbf{d}-\mathbf{m}\right)
		\label{Eq_CE_MF_5}
	\end{equation}
	
	(\ref{Eq_CE_MF_5}) holds when $ c=1 $. Thus CE is equivalent to MF. 
\end{proof}

\ifCLASSOPTIONcaptionsoff
  \newpage
\fi

\bibliographystyle{ieeetr}
\bibliography{Reference}

\end{document}